\documentclass[11pt,letterpaper]{article}
\usepackage[margin=1in]{geometry}
\usepackage[T1]{fontenc}
\usepackage[utf8]{inputenc}
\usepackage{lmodern}
\usepackage{microtype}
\usepackage{amsmath,amssymb,amsthm,mathtools}
\usepackage{enumitem}
\usepackage{cite}
\usepackage{xcolor}
\usepackage[colorlinks=true,linkcolor=blue!45!black,citecolor=blue!45!black,urlcolor=blue!45!black]{hyperref}
\usepackage{doi}
\hypersetup{pdftitle={Quantum Advantage of Permutation-Invariant Functions in Communication Complexity},pdfauthor={Yunqi Huang and Zekun Ye}}
\numberwithin{equation}{section}
\newtheorem{theorem}{Theorem}[section]
\newtheorem{lemma}[theorem]{Lemma}
\newtheorem{proposition}[theorem]{Proposition}
\newtheorem{corollary}[theorem]{Corollary}
\newtheorem*{profilelemma}{Lemma~\ref{lem:profile}}
\newtheorem*{mainrestatement}{Theorem~\ref{thm:main}}
\newtheorem*{inverserestatement}{Theorem~\ref{thm:inverse}}

\newtheorem*{separationsrestatement}{Theorem~\ref{thm:separations}}
\newtheorem*{logrestatement}{Theorem~\ref{thm:log-obstruction}}
\newtheorem*{polynomialrestatement}{Corollary~\ref{cor:no-pure-polynomial}}
\newtheorem*{scalerestatement}{Lemma~\ref{lem:scale}}
\newtheorem*{adjacentrestatement}{Lemma~\ref{lem:adjacent-grid}}
\newtheorem*{alphabetrestatement}{Proposition~\ref{prop:alphabet-upper}}
\newtheorem*{quantumupperrestatement}{Theorem~\ref{thm:quantum-upper}}
\newtheorem*{quantumrestatement}{Theorem~\ref{thm:quantum-characterization}}
\theoremstyle{definition}
\newtheorem{definition}[theorem]{Definition}
\newcommand{\Rpub}{R^{\mathrm{pub}}}
\newcommand{\Qstar}{Q^{\ast}}
\newcommand{\E}{\mathbb E}
\newcommand{\Prb}{\mathbb P}
\newcommand{\F}{\mathbb F}
\newcommand{\calE}{\mathcal E}
\newcommand{\calD}{\mathcal D}
\DeclareMathOperator{\supp}{supp}
\DeclareMathOperator{\Var}{Var}
\DeclareMathOperator{\IP}{IP}
\DeclareMathOperator{\DISJ}{DISJ}
\DeclareMathOperator{\bin}{bin}
\DeclareMathOperator{\Maj}{Maj}
\newcommand{\fall}[2]{(#1)_{\underline{#2}}}

\title{Quantum Advantage of Permutation-Invariant Functions
in Communication Complexity\thanks{The authors are listed in alphabetical order.}}
\author{Yunqi Huang\thanks{Email: \texttt{huangyunqi29@gmail.com}}\qquad
Zekun Ye\thanks{Email: \texttt{yezekun@fzu.edu.cn}}\\[6pt]
\small $^{\dagger}$Quantum Science Center of Guangdong-Hong Kong-Macao Greater Bay Area,\\
\small Shenzhen 518045, China\\[5pt]
\small $^{\dagger}$College of Physics and Optoelectronic Engineering, Shenzhen University,\\
\small Shenzhen 518060, China\\[5pt]
\small $^{\ddagger}$College of Computer and Data Science, Fuzhou University,\\
\small Fuzhou 350108, China}
\date{}

\begin{document}
\setcounter{footnote}{0}
\maketitle
\begin{abstract}
We study how symmetry affects quantum advantage in two-party communication complexity. For any partial function $f$ over a fixed alphabet of size $q$ that is
invariant under simultaneous coordinate permutations, we prove that its
public-coin randomized communication complexity satisfies
\[
R^{\mathrm{pub}}(f)
=O_q\!\left(t\log\left(2+\frac{n}{t}\right)\right),
\qquad
t=\min\{n,Q^{\ast}(f)^2\},
\]
where $Q^{\ast}(f)$ denotes entanglement-assisted quantum communication
complexity, $n$ is the input length, and the implied constant depends
only on $q$. This bound is tight for set disjointness, for which
$Q^{\ast}(\mathrm{DISJ}_n)=\Theta(\sqrt n)$ and
$R^{\mathrm{pub}}(\mathrm{DISJ}_n)=\Theta(n)$. The exponent of the logarithmic factor cannot be reduced by any fixed
positive amount: for every fixed $\varepsilon\in(0,1)$, there are binary
permutation-invariant partial functions with quantum communication
complexity $O_\varepsilon(\log\log n)$ and randomized communication
complexity $\Omega_\varepsilon((\log n)^{1-\varepsilon})$. We also characterize quantum communication complexity by a combinatorial parameter up to a single logarithmic factor. These results extend the binary-alphabet result of Guan et al.~\cite{GHYY} and reduce its logarithmic overhead.

Furthermore, growing alphabets and graph symmetries admit exponential quantum advantages. For every fixed $\varepsilon\in(0,1)$, we exhibit permutation-invariant partial functions on length-$n$ strings over an $n$-symbol alphabet whose quantum communication complexity is $O_\varepsilon(\log n)$, while their randomized communication complexity is $\Omega_\varepsilon(n^{1-\varepsilon})$. We also exhibit graph-invariant partial functions on $v$-vertex graphs whose quantum communication complexity is $O_\varepsilon(\log v)$, while their randomized communication complexity is $\Omega_\varepsilon(v^{2-\varepsilon})$. All separation protocols require no prior entanglement. Thus, permutation-invariant functions over fixed alphabets admit at most quadratic quantum speedup up to a logarithmic factor, whereas growing alphabets and graph symmetries permit exponential quantum--classical separations.
\end{abstract}

\clearpage
\begingroup
\normalsize
\setlength{\parskip}{0pt}
\tableofcontents
\endgroup
\clearpage

\section{Introduction}\label{sec:introduction}
Understanding which properties of a computational problem allow a large quantum advantage is a central question in quantum complexity theory. The query model provides several structural restrictions on such advantages. Beals, Buhrman, Cleve, Mosca, and de Wolf~\cite{BBCMW01} proved that deterministic and bounded-error quantum query complexities are polynomially related for total Boolean functions. This conclusion need not hold for partial functions. For these functions, symmetry provides another restriction on quantum advantage. Aaronson and Ambainis~\cite{AA14} established a polynomial relation between randomized and quantum query complexities under invariance with respect to both coordinate permutations and alphabet relabellings. Chailloux~\cite{Chailloux19} removed the alphabet-relabelling requirement and proved a cubic relation under coordinate-permutation invariance alone. Ben-David, Childs, Gily\'en, Kretschmer, Podder, and Wang~\cite{BDCGKPW} extended polynomial simulation results to graph and hypergraph symmetries in the adjacency-matrix query model.

Communication complexity asks a related question about distributed inputs. In the classical and quantum models introduced by Yao~\cite{Yao79,Yao93}, Alice and Bob have unrestricted local access to their respective inputs, and the resource being measured is the number of bits or qubits exchanged. Large quantum advantages are possible in this setting: Raz~\cite{Raz99} gave an exponential separation for a partial function, and Klartag and Regev~\cite{KR11} obtained such a separation even when the quantum protocol sends only one message and the randomized protocol may interact. These separations, together with the query simulation results above, motivate the question: \textbf{\emph{how does symmetry shape the quantum advantage in the communication complexity model?}}

We study partial functions $f:[q]^n\times[q]^n\to\{-1,+1,*\}$ satisfying
$$
f(\pi x,\pi y)=f(x,y) \qquad\text{for every }\pi\in S_n.
$$
The same coordinate permutation acts on both inputs. The equality includes the undefined value $*$, so the promise must also be invariant. Computing such a function may still require communication, because its value can depend on joint counts that neither party knows locally. The information preserved by simultaneous coordinate permutations is exactly the \emph{joint type}
$$
T_{ab}\coloneqq \bigl|\{i\in[n]:(x_i,y_i)=(a,b)\}\bigr|, \qquad a,b\in[q].
$$
Alice knows the row sums of this table and Bob knows its column sums. After exchanging these marginal histograms, they must still learn enough about the joint type to determine its label.

The binary case was studied by Ghazi, Kamath, and Sudan~\cite{GKS16} and, for quantum communication, by Guan, Huang, Yao, and Ye~\cite{GHYY}. Moving beyond binary inputs changes the structure of the problem. Once the two Hamming weights are fixed, a binary joint type has only one free parameter, the intersection count. Over a $q$-symbol alphabet, fixed-margin tables can instead have $(q-1)^2$ degrees of freedom. Our positive results (Theorems~\ref{thm:main}
and~\ref{thm:quantum-characterization}) extend the results in~\cite{GHYY} from the binary case to every fixed alphabet, with improved
logarithmic factors. The subsequent separation constructions and exponent obstructions (Theorems~\ref{thm:separations} and~\ref{thm:log-obstruction} and Corollary~\ref{cor:no-pure-polynomial}) further clarify the scope and limitations of these fixed-alphabet results.

\subsection{Related Work}
\paragraph{Symmetric Communication Problems.}
Several subclasses of permutation-invariant functions have long served as examples of quantum communication advantages and lower-bound methods. Set disjointness has linear randomized communication complexity~\cite{Razborov92}, whereas quantum search gives a quadratic improvement~\cite{BCW98,AASpatial05}. Razborov~\cite{Razborov03} characterized, up to logarithmic factors, the quantum communication complexity of all total predicates depending only on the intersection size $|x\cap y|$, with lower bounds that allow prior entanglement. Suruga~\cite{Suruga} subsequently obtained tight upper bounds for this class, treating the models with and without prior entanglement separately.

Hamming-distance predicates provide another important subclass. Huang, Shi, Zhang, and Zhu~\cite{HSZZ06} studied threshold predicates, giving lower bounds for interactive quantum protocols with prior entanglement and upper bounds for public-coin randomized simultaneous-message protocols. Zhang and Shi~\cite{ZS09} characterized, up to polylogarithmic factors, the randomized and entanglement-assisted quantum communication complexities of total functions of the form $D(|x\oplus y|)$. For the Gap-Hamming-Distance problem, Chakrabarti and Regev~\cite{CR12} proved an $\Omega(n)$ randomized lower bound when the distance is promised to be at most $n/2-\sqrt n$ or at least $n/2+\sqrt n$. 

\paragraph{General Permutation Invariance.}
Ghazi, Kamath, and Sudan~\cite{GKS16} introduced the general class of binary permutation-invariant communication functions and gave a combinatorial measure describing their randomized complexity up to polynomial factors and an additive logarithmic term. They also observed that arbitrary communication inputs can be stored locally in Hamming weights, at the cost of an exponential increase in input length. This observation is relevant to the scope of any simulation: permutation invariance alone cannot eliminate all dependence on the input length.

Guan, Huang, Yao, and Ye~\cite{GHYY} gave a nearly tight characterization of quantum communication complexity and a nearly quadratic quantum--classical separation for binary permutation-invariant partial functions. Questions about larger alphabets appear in both~\cite{GKS16,GHYY}; Guan et al.~\cite{GHYY} also propose graph inputs. We address these directions by separating fixed alphabets from growing alphabets and by constructing graph-invariant communication problems. The precise binary bounds and parameter comparison appear in Section~\ref{sec:comparison}.

\paragraph{Polynomial and Matrix Methods.}
Our lower bound draws on a line of work relating communication complexity to polynomial approximation. Paturi~\cite{Paturi92} determined the approximate degree of symmetric total Boolean functions, and Nayak and Wu~\cite{NW} established degree lower bounds for distinguishing two promised Hamming weights. Razborov's analysis~\cite{Razborov03} uses symmetry and polynomial approximation to prove quantum communication lower bounds. Sherstov's pattern matrix method~\cite{Sherstov11} gives a general way to transfer approximate-degree lower bounds to quantum communication, including protocols with prior entanglement. For larger-alphabet joint types, we combine the multislice contraction theorem of Braverman, Khot, Lifshitz, and Minzer~\cite{BKLM} with acceptance-matrix factorization bounds from the factorization-norm approach~\cite{LinialShraibman09,SS}. The resulting approximation is uniform over a region of fixed-margin types; this uniformity is needed for the subsequent one-variable degree argument.

\paragraph{Forrelation and Communication Separations.}
Forrelation and its extensions supply the communication problems used in our counterexamples. Aaronson and Ambainis~\cite{AAForrelation18} proved a nearly optimal randomized query lower bound for basic Forrelation. Nearly optimal query separations were obtained independently by Bansal and Sinha~\cite{BS}, using explicit $k$-Forrelation, and by Sherstov, Storozhenko, and Wu~\cite{SSW}, using $k$-Rorrelation. Both Bansal and Sinha~\cite[Corollary~1.6(a)]{BS} and Sherstov, Storozhenko, and Wu~\cite{SSW} also obtained quantum versus randomized \emph{communication} separations of $O_\varepsilon(\log m)$ versus $\Omega_\varepsilon(m^{1-\varepsilon})$ for each fixed $\varepsilon\in(0,1)$. Query-to-communication lifting, including the inner-product lifting theorem of Chattopadhyay, Filmus, Koroth, Meir, and Pitassi~\cite{CFKMP19}, is a tool for passing between these models.

Girish, Raz, and Tal~\cite{GRT} gave a communication separation for an XOR lift of Forrelation, using an entanglement-assisted quantum simultaneous-message protocol with efficient players. Girish, Sinha, Tal, and Wu~\cite{GSTW} strengthened its randomized lower bound through improved Fourier-growth estimates. Our constructions use local encodings to impose permutation or graph invariance on hard communication problems. The contribution of these encodings is the symmetry and the control of local histograms or graph isomorphism types; the underlying communication separations come from the cited work.

\subsection{Our Results}

Let $\Rpub(f)$ be bounded-error public-coin randomized communication complexity, $Q(f)$ be quantum communication complexity without prior entanglement, and $\Qstar(f)$ be quantum communication complexity with arbitrary prior entanglement. The unassisted model $Q$ does not include free shared randomness. All errors are at most $1/3$, and local computation is unrestricted. The precise conventions are given in Section~\ref{sec:prelim}. We take $n\ge2$ throughout. For comparisons of global communication complexities, we assume $\Qstar(f)\ge1$. In the positive results, $q$ is an arbitrary fixed constant; constants indexed by $q$ may depend on $q$, but not on the function or the input length. 

\paragraph{Classical Simulation for Fixed Alphabets.}
Our main theorem gives a classical simulation whose quadratic dependence is on the entanglement-assisted quantum complexity.

\begin{theorem}[Fixed-Alphabet Simulation]\label{thm:main}
For each $q\ge2$, there is a constant $C_q>0$ such that every partial function $f:[q]^n\times[q]^n\to\{-1,+1,*\}$ invariant under simultaneous coordinate permutations and satisfying $\Qstar(f)\ge1$ obeys
\begin{equation}\label{eq:main}
\Rpub(f)\le C_q\,t\log(2+n/t), \qquad t\coloneqq\min\{n,\Qstar(f)^2\}.
\end{equation}
\end{theorem}

This theorem follows by combining the fixed-margin classical upper
bound in Theorem~\ref{thm:classical} with the quantum lower bound in
Theorem~\ref{thm:inverse}. The full proof is given in Section~\ref{sec:consequences}. Theorem~\ref{thm:main} is tight for binary set disjointness. Let $\DISJ_n(x,y)=+1$ when $x,y\in\{0,1\}^n$ have disjoint supports, and $-1$ otherwise. This function is permutation-invariant and satisfies $\Qstar(\DISJ_n)=\Theta(\sqrt n)$ and $\Rpub(\DISJ_n)=\Theta(n)$~\cite{Razborov92,Razborov03,AASpatial05}. Thus $t=\Theta(n)$, and Theorem~\ref{thm:main} gives the optimal $O(n)$ randomized communication bound for this family.

For arbitrary permutation-invariant functions, the theorem implies $\Rpub(f)=O_q(\Qstar(f)^2\log n)$. Thus it both extends the binary simulation to arbitrary fixed alphabets and reduces the logarithmic overhead.

\paragraph{Quantum Characterization.}
For a fixed-margin restriction $g=f_{\rho,\sigma}$, let $h_g$ be the separation parameter of Definition~\ref{def:separation}. When both output labels occur, this parameter measures the minimum separation between oppositely labelled joint types, using square-root count differences normalized within each cut event. Smaller values of $h_g$ require more accurate count estimates. The direct quantum protocol in Appendix~\ref{sec:quantum-upper},
together with the quantum lower bound in Theorem~\ref{thm:inverse},
yields the characterization in
Theorem~\ref{thm:quantum-characterization}.

\begin{theorem}[Quantum Characterization]\label{thm:quantum-characterization}
For a permutation-invariant partial function $f$, define
$$
M(f)\coloneqq\max_{\rho,\sigma}h_{f_{\rho,\sigma}}^{-1},
$$
where a fixed-margin restriction with at most one promised label has $h_{f_{\rho,\sigma}}=1$. The maximum is over all pairs of nonnegative integer histograms $\rho,\sigma$ of total $n$. For constants $c_q,C_q>0$ depending only on $q$,
\begin{equation}\label{eq:quantum-characterization}
c_qM(f)\le\Qstar(f)+1\le Q(f)+1 \le C_q\min\{n,M(f)\log n\}.
\end{equation}
\end{theorem}

Proposition~\ref{prop:membership} shows that the logarithmic factor relative to $h_g^{-1}$ is necessary over a range of binary fixed-margin examples. The theorem follows from Theorems~\ref{thm:quantum-upper} and~\ref{thm:inverse}; the bounds and the displayed consequence are proved in Section~\ref{sec:consequences}.

\paragraph{Necessity of the Input-Length Dependence.}
The input-length factor in Theorem~\ref{thm:main} cannot simply be discarded. Applying the encoding observation of~\cite{GKS16} to a known communication separation gives the explicit families in Corollary~\ref{cor:no-pure-polynomial} below. Their randomized communication complexity grows faster than any fixed polynomial in their quantum communication complexity. Hence some dependence on the input length is necessary: there is no simulation polynomial in quantum communication complexity alone, even when the quantum protocol uses neither prior entanglement nor shared randomness. This shows that the polynomial relation between randomized and quantum communication complexity conjectured by Guan et al.~\cite{GHYY} requires an additional dependence on the input length, even for binary permutation-invariant partial functions. The construction stores the underlying communication inputs in the two Hamming weights; every fixed-margin restriction is constant.

\paragraph{Growing Alphabets and Graph Invariance.}
The next constructions give polynomial randomized lower bounds in
the string length or vertex count, while keeping the quantum cost
logarithmic. Their structural feature is that the local histograms or
graph isomorphism types can be fixed. A \emph{rigid} graph is one whose
only automorphism is the identity.

\begin{theorem}[Growing Alphabets and Graph Invariance]
\label{thm:separations}
For every fixed $\varepsilon\in(0,1)$, there are infinitely many sizes with the following properties.
\begin{enumerate}[label=(\roman*)]
\item There is a permutation-invariant partial function $f_n:[n]^n\times[n]^n\to\{-1,+1,*\}$ such that every promised local input contains every symbol exactly once, and
$$
Q(f_n)=O_\varepsilon(\log n), \qquad \Rpub(f_n)=\Omega_\varepsilon(n^{1-\varepsilon}).
$$
\item There is a partial function $g_v$ of two labelled graphs on $v$ vertices, invariant under simultaneous vertex relabelling, such that
$$
Q(g_v)=O_\varepsilon(\log v), \qquad \Rpub(g_v)=\Omega_\varepsilon(v^{1-\varepsilon}).
$$
Every promised local graph is a labelled copy of one fixed rigid tree of maximum degree three.
\item There is a partial function $\widetilde g_v$ of two connected simple graphs on $v$ vertices, invariant under independent relabelling of the two local graphs, such that
$$
Q(\widetilde g_v)=O_\varepsilon(\log v), \qquad \Rpub(\widetilde g_v)=\Omega_\varepsilon(v^{2-\varepsilon}).
$$
\end{enumerate}
\end{theorem}

Part~(i) rules out a uniform upper bound polynomial in $\Qstar(f)$, $\log n$, and $\log q$ when the alphabet may grow with $n$. Part~(ii) fixes even the entire local graph isomorphism type: the hard input is stored in the relative labelling of two copies of the same rigid tree. These two constructions show that fixing each party's local structure does not remove the communication difficulty of their relative alignment.

Part~(iii) has a different role. Theorem~\ref{thm:general-graphs} embeds \emph{any} $m$-bit communication problem into connected graphs on $\Theta(\sqrt m)$ vertices, preserving all three communication models. The separation follows by applying this embedding to the communication problem of Theorem~\ref{thm:forrelation}. Its larger exponent reflects the $\Theta(v^2)$ bits available in varying graph isomorphism types; the fixed-tree construction instead isolates the relative-labelling mechanism. 


\paragraph{Separate Obstructions to Improving the Exponents.}
Set disjointness already rules out replacing the power two by $2-\eta$, for any fixed $\eta>0$, while retaining any fixed polylogarithmic factor in $n$~\cite{Razborov92,AASpatial05}. We establish a complementary obstruction for the input-length logarithm.

\begin{theorem}[Logarithmic Exponent Obstruction]
\label{thm:log-obstruction}
For every fixed $p\ge0$ and $\delta\in(0,1]$, there is no constant $C$ such that every binary permutation-invariant partial function $f$ with $\Qstar(f)\ge1$ satisfies
$$
\Rpub(f)\le C\Qstar(f)^p(\log n)^{1-\delta}.
$$
\end{theorem}

The construction used to prove Theorem~\ref{thm:log-obstruction} also yields the following explicit separation. Its proof is given in Section~\ref{sec:weight-obstruction}.

\begin{corollary}
\label{cor:no-pure-polynomial}
For every fixed $\varepsilon\in(0,1)$, there is a family of binary partial functions, invariant under independent coordinate permutations, such that along an unbounded sequence of lengths,
$$
Q(f_n)=O_\varepsilon(\log\log n),\qquad \Rpub(f_n)=\Omega_\varepsilon((\log n)^{1-\varepsilon}).
$$
Consequently, for every fixed $p\ge0$,
$$
\frac{\Rpub(f_n)}{(Q(f_n)+1)^p}\longrightarrow\infty.
$$
The same conclusion holds with $Q$ replaced by $\Qstar$.
\end{corollary}


\subsection{Comparison with Prior Binary-Alphabet Results}
\label{sec:comparison}

For binary permutation-invariant partial functions, Guan, Huang, Yao, and Ye~\cite{GHYY} prove that
$$
\Rpub(f)=O\!\left(
\Qstar(f)^2\log^2 n\log\log n+\log n
\right)
$$
in our notation. Theorem~\ref{thm:main} extends this simulation to every fixed alphabet and sharpens its overhead to~\eqref{eq:main}. The input-length dependence remains necessary as shown by Corollary~\ref{cor:no-pure-polynomial}.

For the upper bound, Proposition~\ref{prop:binary-parameter-comparison} shows our cut-based parameter and the jump measure underlying the characterization of~\cite{GHYY} are equivalent for the binary fixed-margin restrictions. Our sharper binary upper bounds are obtained through improved estimation protocols. Specifically, Guan et al.~\cite{GHYY} combine Set-Inclusion tests and binary search in the design of the classical algorithm. We instead estimate joint-type counts directly. Recovering sampled disagreements in a single batch yields our refined classical bound. Building on this sampling-and-recovery framework, we combine coherent sampling with amplitude estimation to obtain the quantum upper bound. Furthermore, the cut-based formulation further extends the parameter to multidimensional joint types over larger alphabets. The joint type is no longer determined by a single intersection count.
Our protocols jointly estimate multiple counts at cut-dependent precision, while preserving the refined communication bounds.

The main challenge in extending the quantum lower bound to larger alphabets is the multidimensional structure of fixed-margin joint types. An alternating cycle may contain several low-count edges, which complicates the construction of the interval needed for polynomial approximation. The reduction to Exact Set-Inclusion and the pattern-matrix argument in~\cite{GHYY} do not supply the interval needed for such a move. Our maximum-weight spanning tree reduction produces fundamental cycles with controlled counts, after which Lemma~\ref{lem:profile} supplies one polynomial on the whole required region, including types outside the promise. The following overview explains the two steps of this reduction.

\subsection{Proof Overview}

The fixed-alphabet simulation and quantum characterization are based on the same separation parameter for joint types. Classical sampling and coherent quantum estimation give two protocols for learning the relevant counts. The quantum lower bound requires a geometric reduction from multidimensional tables to an integer interval, together with a polynomial approximation that is uniform on that interval.

\paragraph{A Separation Parameter for Joint Types.}
View a joint type as a weighted complete bipartite graph, with one vertex for each row symbol and one for each column symbol. Each bipartition of these vertices defines an event
$$
E\coloneqq\{(a,b):\alpha(a)\ne\beta(b)\}, \qquad \alpha,\beta:[q]\to\{0,1\}.
$$
The coordinates whose symbol pairs belong to $E$ are exactly the disagreements between the locally encoded strings $\alpha(x)$ and $\beta(y)$. Let $\calE$ denote the family of nonempty events obtained this way. Its size depends only on $q$. For joint types $T,U$ with the same row and column margins, define
$$
h(T,U)\coloneqq\max_{E\in\calE} \left(\frac{\sum_{e\in E}(\sqrt{T_e}-\sqrt{U_e})^2}{T(E)+U(E)}\right)^{1/2},
\qquad T(E)\coloneqq\sum_{e\in E}T_e,
$$
with a zero denominator contributing zero. On a fixed-margin restriction, $h_f$ is the minimum of this quantity over oppositely labelled promised types, or one if no such pair exists (Definition~\ref{def:separation}). For a fixed-margin function $f$, we prove
\begin{align}
& \Rpub(f) \le O_q\!\left(u\log(2+n/u)\right), &u&\coloneqq\min\{n,h_f^{-2}\},
\label{eq:intro-classical}\\
& \Qstar(f)+1 \ge\Omega_q(h_f^{-1}).
\label{eq:intro-quantum}
\end{align}
For general inputs, the parties first exchange their histograms at cost $O_q(\log n)$, which is absorbed by the bound in Theorem~\ref{thm:main}. Combining the two inequalities on each resulting restriction proves the simulation theorem.

\paragraph{Classical Count Estimation by Batch Recovery.}
For each cut event, the parties estimate its disagreement count and subsample coordinates, aiming for $s=\Theta_q(h_f^{-2})$ retained disagreements. A public random binary linear map allows them to recover the disagreement vector from their local sketches. In the sparse regime, a fixed recovery capacity of $O_q(s)$ gives communication
$$
O\!\left(\log\sum_{j\le O_q(s)}\binom nj\right) =O_q\!\left(s\log(2+n/s)\right).
$$
For larger $s$, sending the full input suffices. The parties then learn the symbol pairs at the recovered positions. Rescaled ideal sample counts have expected squared Hellinger error $O_q(1/s)$. A union bound over the constant-size cut family gives estimates that separate the true type from every opposite-label type. Recovery failures and overflows are accounted for separately, while the fixed capacity bounds communication on every execution. This proves~\eqref{eq:intro-classical}.

\paragraph{Quantum Count Estimation by Coherent Sampling.}
For each cut event with a positive disagreement count, the quantum protocol samples coordinates at a rate that makes the expected number of retained disagreements a constant. The probabilities of an empty disagreement set and of a singleton disagreement with ordered symbol pair $e$ determine $T_e$ through their ratio. Amplitude estimation estimates these probabilities in square-root distance using $O_q(h_f^{-1})$ coherent calls. We fix finite tables of random tapes, each specifying the sampled coordinates and the random matrix used for recovery. An $O_q(\log n)$-bit index selects a tape from the relevant table, and the reversible implementation accounts for all communication of the index, messages, and control qubits. Between coherent calls, the required reflections act only on Alice's registers. This gives an unassisted protocol with cost $O_q(h_f^{-1}\log n)$, requiring neither prior entanglement nor free shared randomness.

\paragraph{Reducing the Lower Bound to One Alternating Cycle.}
For the quantum lower bound, fix oppositely labelled types $T,U$ and symmetrize a protocol, so its acceptance probability depends only on the joint type. Since $T$ and $U$ have the same margins, $U-T$ is an integer circulation on the bipartite symbol graph. We decompose $U-T$ into $O_q(1)$ alternating cycle moves conformally, meaning that each affected entry changes in the direction prescribed by $U-T$. The partial sums therefore stay entrywise between $T$ and $U$, giving a path of legal integer types between the endpoints. Telescoping acceptance probabilities selects one cycle move of size $g$ with a constant acceptance gap depending only on $q$.

The selected cycle can have several low-count edges, so its affine line may contain only the two target integer points. To construct a longer interval of valid types, we assign each edge of the complete bipartite symbol graph a weight equal to the smaller of its two counts in the selected types, and choose a spanning tree of maximum total weight. For a non-tree edge of the selected cycle with common count $B$, let $\lambda$ be the smallest common count on its tree path. The cut created by deleting a minimum-weight path edge yields
$$
(B+g)\lambda=\Omega_q\!\left(g^2/h(T,U)^2\right).
$$
This product retains both count scales needed by the degree argument. Decomposing the selected move into fundamental cycles of the tree and telescoping again produces a cycle with a constant acceptance gap and large counts on its remaining edges. Removing a fixed common table then gives an integer interval of reduced types, with the two target types at adjacent points. Throughout this interval, a connected spanning path has entries bounded below by a constant fraction of the reduced length.

\paragraph{Uniform Approximation on Fixed-Margin Types.}
Lemma~\ref{lem:profile}, stated in Section~\ref{sec:quantum} and proved in Appendix~\ref{sec:polynomial}, approximates the type-averaged acceptance probability of a $c$-qubit protocol by one polynomial of degree $O_q(c+\log(1/\varepsilon))$ throughout the required region. We combine multislice contraction~\cite{BKLM} with the acceptance-matrix factorization underlying the $\gamma_2$ lower-bound method for entanglement-assisted communication~\cite{LinialShraibman09,SS}. The factorization bounds the normalized trace norm by $4^c$, and exponential contraction then controls the discarded high degrees; averaging the remaining coordinate juntas gives polynomials through falling-factorial formulas.

Only a connected spanning subgraph needs large counts. The projections and expansion coefficients depend on the fixed margins and protocol, so the same polynomial works as the type varies, including outside the promise. Substitution along the constructed integer interval and the degree bound of Nayak and Wu~\cite{NW} give
$$
c+1=\Omega_q\!\left(\frac{\sqrt{(B+g)\lambda}}{g}\right) =\Omega_q\!\left(h(T,U)^{-1}\right),
$$
proving~\eqref{eq:intro-quantum}.

\paragraph{Local Encodings for Permutation and Graph Symmetries.}
Theorem~\ref{thm:forrelation} combines the explicit query separation of Bansal and Sinha~\cite{BS}, inner-product lifting~\cite{CFKMP19}, and an XOR extension to construct the base XOR communication problem. Local encodings then preserve its communication complexity while imposing the required symmetries. In the permutation encoding, each bit is the ordering of two distinct symbols. The promise makes the position pairs coincide; a common permutation complements both decoded bits whenever it reverses a pair, preserving their XOR. For the rigid-tree encoding, each party recovers its unique local labelling, and simultaneous relabelling preserves their relative permutation. Both constructions admit local reductions in both directions. The general-graph encoding instead orders data vertices through an isomorphism-invariant backbone and stores bits in their mutual edges. Finally, the Hamming-weight encoding of Ghazi, Kamath, and Sudan~\cite{GKS16} embeds the same XOR communication problem into binary permutation-invariant inputs of exponentially larger length. This yields the input-length obstructions in Section~\ref{sec:weight-obstruction}.

\subsection{Organization}
Section~\ref{sec:prelim} introduces the communication models, joint types, and their separation parameter. Section~\ref{sec:classical} proves the classical upper bound and states its quantum counterpart. Section~\ref{sec:quantum} develops the quantum lower bound and combines these bounds to prove the classical simulation and quantum characterization. Section~\ref{sec:counterexamples} gives the symmetry-preserving encodings, and Section~\ref{sec:limits} establishes the exponent obstructions and alphabet-dependence bounds. Section~\ref{sec:conclusion} summarizes the results and remaining simulation questions. Appendix~\ref{sec:quantum-upper} gives the complete quantum upper-bound proof. Appendix~\ref{sec:polynomial} proves the uniform polynomial approximation lemma. Appendix~\ref{sec:auxiliary} contains scale estimation, batch recovery, sampled-count error, the adjacent-integer degree estimate, alphabet-constant tracking, and the binary-parameter comparison.

\section{Preliminaries}\label{sec:prelim}

Throughout, $n\ge2$ and $q\ge2$ are integers. The symbol $\coloneqq$ denotes a definition. All logarithms are base two, except that $\ln$ denotes the natural logarithm. For a positive integer $m$, let $[m]\coloneqq\{1,\ldots,m\}$. Let $S_n$ denote the symmetric group on $[n]$. A \emph{string} over $[q]$ is a finite sequence of symbols from $[q]$; a length-$n$ input string is an element of $[q]^n$. For a binary string $w\in\{0,1\}^L$, let $|w|\coloneqq\sum_{i=1}^L w_i$ denote its Hamming weight. The field with two elements is denoted by $\mathbb F_2$. The symbol $\oplus$ denotes bitwise XOR (addition modulo two); indexed occurrences denote repeated XOR, with the empty XOR equal to zero. The symbol $*$ denotes an undefined value of a partial function.

\begin{definition}[Communication Models]\label{def:communication-models}
The complexity $\Rpub(f)$ is the minimum worst-case number of communicated classical bits in an interactive public-coin protocol that computes $f$ with error at most $1/3$ on each promised input. The quantities $Q(f)$ and $\Qstar(f)$ are the corresponding quantum communication complexities in qubits, without prior entanglement and with arbitrary prior entanglement, respectively. Local computation is unrestricted. We take Bob to produce the output; requiring both parties to output changes the costs by at most one bit.
\end{definition}

In particular, $\Qstar(f)\le Q(f)$ and $\Qstar(f)\le\Rpub(f)$.
For the latter inequality, shared entanglement can generate shared random bits, and quantum messages can transmit classical bits. We do not assume free shared randomness in the definition of $Q$. If $\Qstar(f)=0$, no-signalling implies that the function value on promised inputs is determined by Bob's input alone. Hence $\Rpub(f)=Q(f)=0$.

\begin{definition}[Permutation Invariance]\label{def:permutation-invariance}
For $\pi\in S_n$, the string $\pi x$ is obtained by moving the symbol at position $i$ to position $\pi(i)$:
$$
(\pi x)_{\pi(i)}\coloneqq x_i\qquad(i\in[n]).
$$
A partial function $f:[q]^n\times[q]^n\to\{-1,+1,*\}$ is \emph{permutation-invariant} if
$$
f(\pi x,\pi y)=f(x,y) \qquad\text{for every }x,y,\pi.
$$
The equality includes the value $*$, so the promise is invariant as well.
\end{definition}

\begin{definition}[Joint Types and Margins]\label{def:types}
The joint type of $(x,y)$ is the table
$$
T_{ab}\coloneqq|\{i:(x_i,y_i)=(a,b)\}|.
$$
We refer to a joint type simply as a type.
Its row margins are $\rho_a\coloneqq\sum_bT_{ab}$, and its column margins are $\sigma_b\coloneqq\sum_aT_{ab}$. The vectors $\rho\coloneqq(\rho_a)_{a\in[q]}$ and $\sigma\coloneqq(\sigma_b)_{b\in[q]}$ are the histograms of $x$ and $y$, respectively. Equivalently,
$$
\rho_a=|\{i\in[n]:x_i=a\}|, \qquad \sigma_b=|\{i\in[n]:y_i=b\}|.
$$
Each histogram has total $n$. Alice can compute $\rho$ from $x$, and Bob can compute $\sigma$ from $y$, without communication. Let $f_{\rho,\sigma}$ denote the restriction of $f$ to inputs with histograms $\rho$ and $\sigma$, leaving its value undefined on all other inputs.
\end{definition}

Two input pairs have the same type if and only if a simultaneous permutation maps one pair to the other. To prove the nontrivial direction, choose a bijection between the positions of each ordered symbol pair $(a,b)$ in the two input pairs. These bijections exist because the joint types are equal, and together they define a simultaneous coordinate permutation mapping one input pair to the other. Consequently a permutation-invariant function is well-defined on its joint types. A type is \emph{promised} if it is realized by a promised input pair. Its label is the common function value on that type. Theorem~\ref{thm:main} allows the margins to vary across inputs. Our protocol first exchanges the two histograms, reducing the problem to a fixed-margin restriction known to both parties.

Consider the complete bipartite graph with row vertices $R_1,\ldots,R_q$ and column vertices $K_1,\ldots,K_q$, which we call the bipartite symbol graph. A cell $(a,b)$ is the edge $R_aK_b$. For an edge $e=(a,b)$, we also let $T_e\coloneqq T_{ab}$, and use the same convention for other tables. For a bipartition of these $2q$ vertices, let $E$ be its crossing edges. Equivalently,
\begin{equation}\label{eq:cuts}
E\coloneqq\{(a,b):\alpha(a)\ne\beta(b)\}, \qquad \alpha,\beta:[q]\to\{0,1\}.
\end{equation}
Let $\calE$ consist of the nonempty events so obtained. Complementing both binary labellings does not change $E$. Fixing $\alpha(1)=0$ therefore gives
\begin{equation}\label{eq:cut-count}
L\coloneqq|\calE|\le2^{2q-1}-1.
\end{equation}
This does not identify $E$ with its complement.

\begin{definition}[Type-Separation Parameter]\label{def:separation}
For tables $T,U$ with the same margins, define the nonnegative quantities
\begin{equation}\label{eq:h}
h_E(T,U)^2\coloneqq\frac{\sum_{e\in E}(\sqrt{T_e}-\sqrt{U_e})^2}{T(E)+U(E)},
\quad T(E)\coloneqq\sum_{e\in E}T_e,
\qquad h(T,U)\coloneqq\max_{E\in\calE}h_E(T,U).
\end{equation}
A ratio with zero denominator is assigned value zero. 
For a permutation-invariant partial function $g$ whose promised inputs have fixed margins, let $h_g$ be the minimum of $h(T,U)$ over promised types with opposite labels. If no such pair exists, set $h_g\coloneqq1$.
\end{definition}

Since $(\sqrt t-\sqrt u)^2\le t+u$, we have $h(T,U)\le1$. If $T\ne U$, the event containing all cells gives $h(T,U)>0$. The set of types is finite, so $0<h_g\le1$.

We next represent each cut separation $h_E(T,U)$ as the Hellinger distance between two auxiliary probability distributions. This representation lets us use the triangle inequality to show that sufficiently accurate count estimates determine the correct output label.

\begin{definition}[Hellinger Distance]\label{def:hellinger}
Let $E$ be a finite set, and let $\bot$ be a distinguished symbol used below to represent the case of zero total count. Let $P,Q$ be probability distributions whose outcomes are either $\bot$ or pairs $(z,e)$ with $z\ge0$ and $e\in E$, with densities on the continuous components.
Let $p_\bot$ denote $P$'s probability of $\bot$, and let $p_e(z)$ denote its density at label $e$: the probability of $z\in[a,b]$ with label $e$ is $\int_a^b p_e(z)\,dz$. Use $q_\bot,q_e$ analogously for $Q$. Their Hellinger distance is defined by
$$
H(P,Q)^2\coloneqq\frac12\left[(\sqrt{p_\bot}-\sqrt{q_\bot})^2 +\sum_{e\in E}\int_0^\infty(\sqrt{p_e(z)}-\sqrt{q_e(z)})^2\,dz\right].
$$
The integrals are ordinary integrals of densities; the separate $\bot$ term accounts for its point probability.
\end{definition}

The triangle inequality for $H$ is the usual $L^2$ triangle inequality applied to the square-root probabilities and densities.

Let $W=(W_e)_{e\in E}$ be a nonnegative real count vector of total $w$. If $w>0$, define $J_W$ to have probability zero at $\bot$ and density $W_e e^{-wz}$ at $(z,e)$, for $z\ge0$ and $e\in E$. Equivalently, sample an exponential random variable of rate $w$ and independently choose label $e$ with probability $W_e/w$. The exponential random variable has density $we^{-wz}$ for $z\ge0$. If $w=0$, set $J_W\coloneqq\delta_\bot$, the distribution that outputs $\bot$ with probability one. For a table $T$, set $J_T^E\coloneqq J_{(T_e)_{e\in E}}$.

\begin{lemma}[Count Formula for Hellinger Distance]\label{lem:hellinger}
For nonnegative count vectors $W,V$ of totals $w,v$,
\begin{equation}\label{eq:hellinger}
H(J_W,J_V)^2=\frac{\sum_e(\sqrt{W_e}-\sqrt{V_e})^2}{w+v},
\end{equation}
with the same zero-denominator convention. In particular, $H(J_T^E,J_U^E)=h_E(T,U)$.
\end{lemma}
\begin{proof}
If $w,v>0$, integration gives
$$
H(J_W,J_V)^2 =1-\sum_e\sqrt{W_eV_e}\int_0^\infty e^{-(w+v)z/2}\,dz =1-\frac{2\sum_e\sqrt{W_eV_e}}{w+v}.
$$
Expanding the squares gives~\eqref{eq:hellinger}. If exactly one total is zero, the two distributions have disjoint support and both sides of~\eqref{eq:hellinger} equal one. If both totals are zero, both distributions equal $\delta_\bot$.
\end{proof}

These auxiliary distributions are used only in the analysis. The protocol below does not sample an exponential random variable.

\begin{lemma}[A Communication Cost Function]\label{lem:phi}
Let $\Phi_n(s)\coloneqq s\log(2+n/s)$ for $s>0$. Then $\Phi_n$ is increasing, $\Phi_n(Ds)\le D\Phi_n(s)$ for $D\ge1$, and $\Phi_n(s)\ge\max\{s,\log(n+2)\}$ for $s\ge1$.
\end{lemma}
\begin{proof}
The derivative multiplied by $\ln2$ is $\ln(2+n/s)-(n/s)/(2+n/s)$. The function $r\mapsto\ln(2+r)-r/(2+r)$ has value $\ln2$ at zero and derivative $r/(2+r)^2\ge0$. This proves monotonicity. The scaling inequality follows by decreasing $n/s$ to $n/(Ds)$. Finally, $\log(2+n/s)\ge1$ and $\Phi_n(s)\ge\Phi_n(1)=\log(n+2)$.
\end{proof}

\section{Classical and Quantum Upper Bounds}\label{sec:classical}

\subsection{Classical Upper Bound}

We first work with known margins. The type-separation parameter then controls how accurately the parties must estimate the joint type to determine its label. The following bound is obtained by estimating the size of each cut event, sampling its disagreement positions, and recovering the sampled positions in one batch.

\begin{theorem}[Classical Upper Bound]\label{thm:classical}
Suppose both parties know the margins $\rho,\sigma$. For a permutation-invariant partial function $f$ restricted to these margins, let $h_f$ be as in Definition~\ref{def:separation}. Then
$$
\Rpub(f)\le C_q\,u\log(2+n/u), \qquad u\coloneqq\min\{n,h_f^{-2}\}.
$$
\end{theorem}

\paragraph{Estimating the Number of Disagreements}

Both estimation and recovery use the following collision bound for random linear maps, the binary universal hash family of Carter and Wegman~\cite{CW79}.

\begin{lemma}[Random Linear Hashing, {\cite{CW79}}]
\label{lem:linear-hash}
Let $m\ge1$ be an integer, and let $H\in\F_2^{m\times n}$ have independent uniform binary entries. For every fixed pair of distinct vectors $u,v\in\F_2^n$,
$$
\Prb[Hu=Hv]=2^{-m}.
$$
In particular, setting $v=0$ gives $\Prb[Hu=0]=2^{-m}$ for every fixed nonzero $u\in\F_2^n$.
\end{lemma}

The next protocol combines this fingerprint with nested coordinate samples. We include its analysis because we need a one-sided zero case and both parties obtain the same estimate $\widetilde k$.

\begin{lemma}[Estimating the Disagreement Count]\label{lem:scale}
For each fixed $\delta\in(0,1)$, there is a constant $\kappa\ge1$ and a public-coin protocol using $O_\delta(\log n)$ bits with the following property. On binary strings with $k$ disagreements it outputs the same value $\widetilde k$ to both parties. If $k=0$, the output is always zero. If $k>0$, then
$$
\Prb[k/\kappa\le\widetilde k\le\kappa k]\ge1-\delta.
$$
Each nonzero output is an integer power of two.
\end{lemma}
The nested-sampling argument, including its one-sided zero case, is proved in Appendix~\ref{sec:aux-scale}.

\paragraph{Recovering the Sampled Disagreement Set}

The following lemma recovers all disagreement positions in one batch by applying Lemma~\ref{lem:linear-hash} to the disagreement vector. It applies to arbitrary binary strings; the sampling step will be specified in the proof of Theorem~\ref{thm:classical}.

\begin{lemma}[Recovering the Disagreement Set]
\label{lem:batch-recovery}
Let $n\ge2$, let $B$ be an integer with $1\le B\le n/2$, and let $\varepsilon\in(0,1)$. There is a public-coin protocol with the following properties. Alice holds $\xi\in\F_2^n$, and Bob holds $\eta\in\F_2^n$. The protocol either reports failure to both parties or gives both parties the same set $\widehat D\subseteq[n]$ of size at most $B$. If the actual disagreement set
$$
D\coloneqq\{i\in[n]:\xi_i\ne\eta_i\}
$$
has size at most $B$, then
$$
\Prb[\text{the protocol returns }\widehat D=D]\ge1-\varepsilon.
$$
For every input and every execution, the communication is at most
$$
O\bigl(B\log(n/B)+\log(1/\varepsilon)\bigr)
$$
bits. The protocol uses only finite public randomness.
\end{lemma}

The protocol sends a linear sketch and returns the index of the unique weight-at-most-$B$ vector consistent with it. The collision analysis and worst-case cost bound, including overflow inputs, are proved in Appendix~\ref{sec:aux-batch}.

\paragraph{Error of Sampled Counts}

We analyze the counts obtained by ideal sampling, in which every retained coordinate whose symbol pair belongs to $E$ is counted. These counts are defined whether or not a communication protocol successfully recovers the retained positions. We will apply the lemma below before accounting for recovery failures.

\begin{lemma}[Error of Sampled Counts]
\label{lem:count-error}
Fix an input pair $(x,y)$ with joint type $T$, an event $E\in\calE$, and a sampling probability $p\in(0,1]$. Let $k\coloneqq T(E)=\sum_{e\in E}T_e$, and suppose $k>0$. Retain each coordinate independently with probability $p$, and let $A$ be the random set of retained coordinates. For each ordered symbol pair $e\in E$, define the ideal sample count
$$
Z_e\coloneqq\bigl|\{i\in A:(x_i,y_i)=e\}\bigr|.
$$
Define also
$$
\widehat T_e\coloneqq Z_e/p, \qquad \widehat k\coloneqq\sum_{e\in E}\widehat T_e, \qquad \widehat J_E\coloneqq J_{(\widehat T_e)_{e\in E}}.
$$
Then
$$
\E H(J_T^E,\widehat J_E)^2 \le\frac{|E|(1-p)}{pk} \le\frac{q^2(1-p)}{pk}.
$$
Here the expectation is over the coordinate sampling for the fixed input and the fixed value of $p$. If $\widehat k=0$, the definition of $J_W$ gives $\widehat J_E=\delta_\bot$. If $p=1$, the estimate is exact.
\end{lemma}

The proof in Appendix~\ref{sec:aux-count} uses the binomial variance of each ideal count and the Hellinger count formula. Its expectation is taken before accounting for recovery failures.

\paragraph{Classical Protocol and Communication Cost}

The two protocols use the same decoder: choose a promised type whose auxiliary distributions are closest to the estimates. The following lemma isolates the deterministic guarantee needed for both protocols.

\begin{lemma}[Decoding from Hellinger Estimates]\label{lem:type-decoder}
Fix a permutation-invariant function $f$ with known margins and both output labels, and let $h\coloneqq h_f$. Let $T$ be the true promised type. Suppose probability distributions $\widehat J_E$, one for each $E\in\calE$, satisfy
$$
\max_{E\in\calE}H(J_T^E,\widehat J_E)\le\varepsilon.
$$
For every promised type $U$ with these margins, set
$$
D(U)\coloneqq\max_{E\in\calE}H(J_U^E,\widehat J_E).
$$
Then $D(T)\le\varepsilon$, whereas $D(U)\ge h-\varepsilon$ for every oppositely labelled $U$. In particular, minimizing $D$ returns the correct label if $\varepsilon<h/2$. More generally, let $\zeta\ge0$ and suppose the approximate scores $\widetilde D(U)$ satisfy
$$
|\widetilde D(U)-D(U)|\le\zeta
$$
for every promised type $U$. Then minimizing $\widetilde D$ returns the correct label whenever $\varepsilon+\zeta<h/2$.
\end{lemma}
\begin{proof}
The assertion for $T$ is the assumed accuracy bound. For an oppositely labelled $U$, Definition~\ref{def:separation} and Lemma~\ref{lem:hellinger} give an event $E$ with $H(J_T^E,J_U^E)\ge h$. The triangle inequality gives
$$
D(U)\ge H(J_U^E,\widehat J_E) \ge H(J_U^E,J_T^E)-H(J_T^E,\widehat J_E) \ge h-\varepsilon.
$$
If $U_*$ minimizes the approximate score, then
$$
D(U_*)\le\widetilde D(U_*)+\zeta \le\widetilde D(T)+\zeta \le D(T)+2\zeta \le\varepsilon+2\zeta<h-\varepsilon.
$$
Thus $U_*$ cannot have the opposite label. Taking $\zeta=0$ proves the exact-score case as well. The set of candidate types is finite, so a minimizer exists.
\end{proof}

We now combine these ingredients to prove Theorem~\ref{thm:classical}.

\begin{proof}[Proof of Theorem~\ref{thm:classical}]
A restriction with at most one promised label needs no communication. Otherwise let $h\coloneqq h_f\in(0,1]$. The function and its known margins fix all parameters below in advance.

\textbf{Parameters and Sampling.}
Let $L\coloneqq|\calE|$, set $\delta\coloneqq1/(1000L)$, and take $\kappa\ge1$ from Lemma~\ref{lem:scale} for this $\delta$. Choose
$$
A_q\coloneqq6400L\kappa q^2,\qquad s\coloneqq\lceil A_qh^{-2}\rceil,\qquad B\coloneqq\lceil\kappa s/\delta\rceil.
$$
The constants $L,\delta,\kappa,A_q$ depend only on $q$. If $B>n/2$, Alice sends her full input and Bob computes the answer exactly; its communication cost is analyzed below. For the case $B\le n/2$, we use the following sampling and recovery procedure.

For each $E\in\calE$, choose encodings $\alpha,\beta$ as in~\eqref{eq:cuts}. The strings $\xi_i=\alpha(x_i)$ and $\eta_i=\beta(y_i)$ disagree at exactly $k=T(E)$ positions. Run the protocol from Lemma~\ref{lem:scale} for every event, obtaining estimates $\widetilde k$ known to both parties. An event with $\widetilde k=0$ receives estimate $\widehat J_E=\delta_\bot$ and is skipped. Otherwise set
$$
p\coloneqq\min\{1,s/\widetilde k\}
$$
and retain each coordinate independently with probability $p$, using public randomness independent of that used to obtain $\widetilde k$. Let $A_E$ denote the set of retained coordinates. This uses finite randomness: every nonzero estimate $\widetilde k$ is a power of two, so a uniform integer in $\{1,\ldots,\widetilde k\}$ can be generated using $\log_2\widetilde k$ random bits. When $p<1$, independently draw such an integer for each coordinate and retain the coordinate if the integer is at most $s$. When $p=1$, retain every coordinate.

Whenever $k/\kappa\le\widetilde k\le\kappa k$, we have 
$$
pk\le\kappa s,\qquad p<1\ \Longrightarrow\ pk\ge s/\kappa.
$$
Indeed, for $p<1$ these follow by substituting $p=s/\widetilde k$; for $p=1$, use $k\le\kappa\widetilde k\le\kappa s$. Only retained disagreements need to be recovered; the total number of retained coordinates can be much larger.

\textbf{Recovering the Retained Disagreements.}
Set all coordinates outside $A_E$ to zero in both encoded strings, leaving the coordinates in $A_E$ unchanged. Denote the resulting strings by $\xi'$ and $\eta'$. Their disagreement vector $d_0\coloneqq\xi'\oplus\eta'$ has support $\{i\in A_E:(x_i,y_i)\in E\}$. Apply the protocol from Lemma~\ref{lem:batch-recovery} with capacity $B$ and error parameter $\delta$. Generate its random binary matrix independently of all random choices used to obtain the estimates $\widetilde k$ and coordinate sampling. Use independent sampling randomness and independent recovery matrices for different events. If the recovery protocol reports failure, both parties halt with a fixed output. Otherwise, let $\widehat D_E$ denote the set returned by the protocol; it is known to both parties and satisfies $|\widehat D_E|\le B$. Alice sends $x_i$ for all $i\in\widehat D_E$, in increasing order of $i$. Bob forms
$$
Z_e^{\mathrm{rec}}\coloneqq \bigl|\{i\in\widehat D_E:(x_i,y_i)=e\}\bigr|,
\qquad
\widehat J_E\coloneqq J_{(Z_e^{\mathrm{rec}}/p)_{e\in E}}.
$$
These estimates are defined even after an incorrect recovery; zero counts produce $\delta_\bot$.

For any fixed $\widetilde k$ satisfying $k/\kappa\le\widetilde k\le\kappa k$, Markov's inequality gives
$$
\Prb[|d_0|>B]\le\frac{pk}{B}\le\frac{\kappa s}{B}\le\delta.
$$
Every such overflow is counted as a bad event, whether or not the protocol detects it. For each fixed sample with $|d_0|\le B$, the independent recovery matrix gives probability at most $\delta$ of a reported failure or an incorrect set, by Lemma~\ref{lem:batch-recovery}.

\textbf{Accuracy of the Ideal Counts.}
Fix values of the estimates $\widetilde k$ satisfying the stated
accuracy bounds and analyze the ideal counts without conditioning on recovery success. The subsequent coordinate samples remain independent at their specified rates. For $k>0$, define the ideal counts and distributions
$$
Z_e\coloneqq\bigl|\{i\in A_E:(x_i,y_i)=e\}\bigr|,
\qquad
\widehat J_E^{\mathrm{ideal}}\coloneqq J_{(Z_e/p)_{e\in E}}.
$$
On correct recovery they equal the recovered counts and distributions. For $k=0$, set $\widehat J_E^{\mathrm{ideal}}=\delta_\bot=J_T^E$. For analysis, all samples can be drawn in advance; samples for events not reached before termination remain unused.

When $p<1$, Lemma~\ref{lem:count-error} and $pk\ge s/\kappa$ give
$$
\E H(J_T^E,\widehat J_E^{\mathrm{ideal}})^2 \le\frac{q^2(1-p)}{pk}\le\frac{\kappa q^2}{s}.
$$
For $p=1$ or $k=0$ the ideal estimate is exact, so the same bound holds. Markov's inequality at squared distance $h^2/64$, followed by a union bound over $L$ events, yields
$$
\begin{aligned}
&\Prb\left[\max_{E\in\calE}H(J_T^E,\widehat J_E^{\mathrm{ideal}})>h/8 \,\middle|\,\text{all estimates satisfy the accuracy bounds}\right]\\
&\hspace{2em}\le\frac{64L\kappa q^2}{sh^2}\le\frac1{100}.
\end{aligned}
$$
The last inequality uses $s\ge6400L\kappa q^2h^{-2}$. The bound holds for every fixed collection of estimates $\widetilde k$ satisfying the accuracy bounds, and therefore after averaging over them.

\textbf{Decoding the Label and Bounding the Error.}
After all recovery steps, Bob minimizes the approximate score of Lemma~\ref{lem:type-decoder} over promised types. He approximates every Hellinger distance to additive error $h/16$; taking a maximum preserves this error. If every estimate $\widetilde k$ satisfies the stated accuracy bounds, all recoveries are correct, and all ideal Hellinger errors are at most $h/8$, the lemma applies with $\varepsilon=h/8$ and $\zeta=h/16$ and guarantees the correct label. Counts and sampling rates are rational, and the distance formula uses arithmetic and square roots, so finite local computation achieves the prescribed precision without exact real comparisons. No union bound over candidate types is needed.

By a union bound, the probability that some estimate $\widetilde k$ violates its accuracy bounds is at most $L\delta$. Conditional on all estimates satisfying these bounds, the probability of any overflow is at most $L\delta$. For each fixed sample collection without overflow, the probability of any recovery failure or incorrect recovery is at most $L\delta$. The preceding ideal-count bound is at most $1/100$, so total error is at
most
$$
3L\delta+\frac1{100}=\frac{13}{1000}<\frac13.
$$
This combines conditional bounds and union bounds only: the bad events need not be independent, and no binomial claim is made after conditioning on successful recovery.

\textbf{Worst-Case Communication.}
We now bound the communication cost in both cases. Let $D_q\coloneqq\kappa/\delta+1$, so $s\le B\le D_qs$. When $B\le n/2$, recovery and symbol transmission for each event cost
$$
O\bigl(B\log(n/B)+B\log q+\log(1/\delta)\bigr) =O_q(\Phi_n(B)),
$$
where $\Phi_n(r)=r\log(2+n/r)$. Lemma~\ref{lem:phi} gives $\Phi_n(B)\le D_q\Phi_n(s)$ and $\log(n+2)\le\Phi_n(s)$. Thus the total cost over all $L=O_q(1)$ events, including the $O_q(\log n)$ bits used to obtain the estimates $\widetilde k$, is $O_q(\Phi_n(s))$. This bound holds on every execution, even when an estimate violates its accuracy bounds, an overflow occurs, or recovery fails or returns an incorrect set. When $B>n/2$, we have $n<2B\le2D_qs$, so sending $n\lceil\log q\rceil=O_q(s)$ bits also costs $O_q(\Phi_n(s))$. If both parties must output, one additional bit suffices.

Finally, $1\le h^{-2}$ and $s\le(A_q+1)h^{-2}$. A further application of Lemma~\ref{lem:phi} gives $\Phi_n(s)\le(A_q+1)\Phi_n(h^{-2})$. If $h^{-2}\le n$, use the constructed protocol; otherwise send the full input and use $\Phi_n(n)=n\log3$. Both parties know $h$ and can make this choice in advance. In either case, with $u=\min\{n,h^{-2}\}$,
$$
\Rpub(f)\le C_q\Phi_n(u)=C_q\,u\log(2+n/u).
$$
All randomness is finite and every branch has bounded communication.
\end{proof}

\subsection{Quantum Upper Bound}\label{sec:quantum-upper-summary}
The same separation parameter admits a quantum upper bound with linear rather than quadratic dependence on $h_g^{-1}$.

\begin{theorem}[Quantum Upper Bound]\label{thm:quantum-upper}
For every fixed $q\ge2$, there is a constant $C_q>0$ such that every fixed-margin restriction $g\coloneqq f_{\rho,\sigma}$ of a permutation-invariant partial function $f:[q]^n\times[q]^n\to\{-1,+1,*\}$ satisfies
\begin{equation}
Q(g)\le C_q\min\{n,h_g^{-1}\log n\}.
\label{eq:qupper-20}
\end{equation}
\end{theorem}

The protocol uses neither prior entanglement nor free shared randomness. It estimates the probabilities of empty and singleton disagreement samples coherently; their ratios recover the joint-type counts. Amplitude estimation and the decoder of Lemma~\ref{lem:type-decoder} yield the bound. Appendix~\ref{sec:quantum-upper} gives the full proof, including finite random-tape tables, communication of table indices, controlled calls, and workspace cleanup. The construction is independent of the quantum lower bound in Section~\ref{sec:quantum}.

\section{Quantum Lower Bound}\label{sec:quantum}

We first prove a lower bound for distinguishing two joint types with the same margins. We then combine it with the classical and quantum upper bounds to prove Theorem~\ref{thm:main} and Theorem~\ref{thm:quantum-characterization}, respectively. The analytic input for the lower bound is Lemma~\ref{lem:profile}, proved in Appendix~\ref{sec:polynomial}.

\begin{theorem}[Quantum Lower Bound for Two Types]\label{thm:inverse}
For each $q\ge2$, there is $a_q>0$ with the following property. Let $T\ne U$ be two types of total $n$ with the same row and column margins. If a protocol using $c$ qubits and arbitrary prior entanglement distinguishes all inputs of type $T$ from all inputs of type $U$ with error at most $1/3$, then
\begin{equation}\label{eq:quantum-inverse}
c+1\ge a_q/h(T,U).
\end{equation}
\end{theorem}

\subsection{Polynomial Degree Bounds for Adjacent Integers}

We first state the univariate degree estimate needed at the end of the argument. Boundedness is required only at integer points.

\begin{definition}[Polynomial Conventions]\label{def:polynomial-conventions}
The \emph{total degree} of a polynomial is the largest sum of exponents in a monomial with a nonzero coefficient. A polynomial is \emph{multilinear} if each variable has exponent at most one in every monomial. For a polynomial $A$ in variables $w_1,\ldots,w_L$, its \emph{multilinearization} $\operatorname{ml}(A)$ is obtained by replacing every positive power $w_i^k$ by $w_i$ and combining equal monomials. Thus
$$
\operatorname{ml}(A)(w)=A(w)\quad(w\in\{0,1\}^L),
\qquad
\deg\operatorname{ml}(A)\le\deg A.
$$
For an odd positive integer $\ell$, define the polynomial
$$
\Maj_\ell(z)\coloneqq \sum_{k=(\ell+1)/2}^{\ell} \binom{\ell}{k}z^k(1-z)^{\ell-k}.
$$
For $z\in[0,1]$, $\Maj_\ell(z)$ is the probability of at least $(\ell+1)/2$ successes in $\ell$ independent Bernoulli trials with success probability $z$.
\end{definition}

\begin{lemma}[Polynomial Degree Lower Bound, {\cite{NW}}]\label{lem:nayak-wu}
Let $L\ge1$ and $0\le r<s\le L$ be integers. Suppose a real polynomial $A$ in $L$ variables satisfies
$$
-\frac13\le A(w)\le\frac43\qquad(w\in\{0,1\}^L),
$$
and
$$
|A(w)|\le\frac13\quad\text{when }|w|=r,
\qquad
|A(w)-1|\le\frac13\quad\text{when }|w|=s.
$$
Choose $u\in\{r,s\}$ maximizing $|u-L/2|$. Then, for an absolute constant $c_{\mathrm{NW}}>0$,
$$
\deg A\ge c_{\mathrm{NW}}\left(\sqrt{\frac L{s-r}}+\frac{\sqrt{u(L-u)}}{s-r}\right).
$$
\end{lemma}

This is the error-$1/3$ specialization of the cited theorem. We use it below at consecutive weights; the amplification and change of variables needed for our integer interval are proved separately.

\begin{lemma}[Degree Needed to Separate Adjacent Integers]\label{lem:adjacent-grid}
Let $a,b$ be nonnegative integers. Suppose a real polynomial $F$ satisfies $|F(j)|\le M$ for all integers $-a\le j\le b+1$, and $|F(1)-F(0)|\ge\Delta>0$. Then
$$
\deg F\ge c_{M/\Delta}\sqrt{(a+1)(b+1)},
$$
where $c_{M/\Delta}>0$ depends only on the indicated ratio.
\end{lemma}
The amplification and change of variables reducing this statement to Lemma~\ref{lem:nayak-wu} are given in Appendix~\ref{sec:aux-grid}.

\begin{samepage}
\subsection{Quantum Lower Bound for One Alternating Cycle}

The following uniform polynomial approximation lemma is the analytic input to the single-cycle lower bound. Its proof appears in Appendix~\ref{sec:polynomial}. A row or column symbol is active if its corresponding margin is positive.

\begin{lemma}[Uniform Polynomial Approximation]\label{lem:profile}
Fix an integer $q\ge2$, $\beta>0$, total length $N\ge1$, and row and column margins with at most $q$ active symbols per side. Let $G$ be a connected spanning bipartite subgraph on the active symbols. For any $c$-qubit protocol with prior entanglement, let $p(S)$ denote its acceptance probability averaged uniformly over all input pairs of joint type $S$. For each $\varepsilon\in(0,1)$, there is a single real polynomial $P$ in the entries of $S$, of total degree at most
$$
K_{q,\beta}\bigl(c+\log(1/\varepsilon)\bigr),
$$
such that $|p(S)-P(S)|\le\varepsilon$ for every integer type with the fixed margins satisfying
$$
S_e\ge\beta N\qquad\text{for every }e\in G.
$$
Entries outside $G$ may be zero. Here $K_{q,\beta}>0$ depends only on $q$ and $\beta$.
\end{lemma}
\end{samepage}

The polynomial is chosen for the fixed protocol and margins, rather than separately for each type. Consequently it controls an entire interval of types, including intermediate types outside the function's promise. Only a connected spanning subgraph needs large counts, which will allow one edge of the cycle below to have a small or zero count.

A simple cycle in the bipartite symbol graph has even length $2r$, with $2\le r\le q$. Assigning alternating signs to its edges gives a table $v=v_+-v_-$, where $v_+,v_-$ are disjoint $0/1$ tables. Every row and column sum of $v$ is zero. Adding $gv$ to a type therefore preserves its margins whenever the resulting entries are nonnegative.

\begin{lemma}[Quantum Lower Bound for an Alternating Cycle]\label{lem:single-cycle}
Suppose two types differ by $g$ moves around a simple alternating cycle, where $g\ge1$ is an integer. Define their common counts to be their entrywise minima. Their common count at one distinguished edge is $B$, and their common counts at all other cycle edges are at least $\Lambda$. Assume $B\le2\Lambda$ and $\Lambda\ge2g$. If a $c$-qubit protocol has a type-averaged acceptance gap of at least $1/(6q^3)$ between the two types, then
\begin{equation}\label{eq:single-cycle-bound}
c+1\ge b_q\frac{\sqrt{(B+g)\Lambda}}{g}
\end{equation}
for a constant $b_q>0$.
\end{lemma}
\begin{proof}
Let $e$ be the distinguished edge and orient the endpoints so that $S^+_e-S^-_e=g$. Writing $C$ for their entrywise minimum and $v=v_+-v_-$ for the alternating cycle gives
$$
S^-=C+gv_-,\qquad S^+=C+gv_+,\qquad C_e=B.
$$
Let $H\coloneqq\lfloor\Lambda\rfloor$ and $B'\coloneqq\min\{B,H\}$. Since $g$ is an integer and $\Lambda\ge2g$, we have $H\ge2g$ and $H\ge\Lambda/2$. Let $A_0$ have count $B'$ at $e$, count $H$ on every other cycle edge, and zero elsewhere. Then $A_0\le C$ entrywise, and
$$
\widehat S^-\coloneqq A_0+gv_-,\qquad
\widehat S^+\coloneqq A_0+gv_+
$$
are nonnegative integer types with identical reduced row and column margins.

To transfer the protocol to these reduced types, fix padding strings of type $C-A_0$: use $(C-A_0)_{ab}$ coordinates of pair $(a,b)$. Each party appends its padding string, the parties apply a common uniform coordinate permutation, and they run the original protocol. The padded endpoint types are $S^-,S^+$. A uniform permutation sends any fixed pair uniformly over its type class, because the action is transitive and all fibers have equal size. Let $p$ and $\widehat p$ denote the original and constructed protocols' type-averaged acceptance functions, respectively. Then
$$
\widehat p(\widehat S^-)=p(S^-),\qquad \widehat p(\widehat S^+)=p(S^+).
$$
The gap is preserved. Padding and permutations are local, and prior entanglement supplies the shared permutation, so communication remains at most $c$ qubits.

If the cycle has $2r$ edges, its reduced total length is
\begin{equation}\label{eq:reduced-length}
N'=B'+(2r-1)H+rg\le2rH+\frac{rH}{2}\le3qH,
\end{equation}
using $B'\le H$, $g\le H/2$, and $r\le q$. Set
\begin{equation}\label{eq:integer-types}
X_z\coloneqq A_0+gv_-+zg(v_+-v_-),\qquad -a\le z\le1+b,
\end{equation}
where $z$ is an integer and
$$
a\coloneqq\left\lfloor\frac{B'}{4g}\right\rfloor, \qquad b\coloneqq\left\lfloor\frac H{4g}\right\rfloor.
$$
These tables have fixed margins and total $N'$. The distinguished entry is at least $B'-ag\ge3B'/4$. Every other edge of $v_+$ has count at least $H-ag\ge H-B'/4\ge3H/4$, and every edge of $v_-$ has count at least $H-bg\ge3H/4$. Thus every $X_z$ is a legal type. For these reduced margins, the active symbols are exactly the vertices of the cycle. Deleting $e$ therefore leaves a connected spanning path $G$ on the active symbols, and throughout the interval
$$
(X_z)_{e'}\ge\frac{3H}{4}\ge\frac{N'}{4q} \qquad(e'\in G).
$$

Apply Lemma~\ref{lem:profile} with $N=N'$, $\beta=1/(4q)$, and $\varepsilon=1/(24q^3)$. It gives one polynomial $P$ satisfying
$$
\begin{gathered}
|\widehat p(X_z)-P(X_z)|\le\varepsilon
\quad(z=-a,\ldots,1+b),\\
\deg P\le K_{q,1/(4q)}\bigl(c+\log(24q^3)\bigr).
\end{gathered}
$$
Since the entries of $X_z$ are affine in $z$, the univariate polynomial $F(z)\coloneqq P(X_z)$ has degree at most $D_q(c+1)$ for a constant $D_q$ depending only on $q$. It satisfies $|F(z)|\le1+\varepsilon\le2$ at every integer in the interval, including types outside the promise. At $X_0=\widehat S^-$ and $X_1=\widehat S^+$,
$$
|F(1)-F(0)|\ge |\widehat p(X_1)-\widehat p(X_0)|-2\varepsilon \ge\frac1{6q^3}-2\varepsilon=\frac1{12q^3}.
$$
Lemma~\ref{lem:adjacent-grid}, with $M=2$ and $\Delta=1/(12q^3)$, therefore gives
\begin{equation}\label{eq:interval-degree}
D_q(c+1)\ge\deg F\ge c'_q\sqrt{(a+1)(b+1)}.
\end{equation}
Here $c'_q>0$ depends only on $M/\Delta=24q^3$.

Finally, $H\ge\Lambda/2$ and $B\le2\Lambda\le4H$ imply $B'\ge B/4$. Using $\lfloor t\rfloor+1\ge\max\{1,t\}$ and $\max\{1,B'/(4g)\}\ge(B'+g)/(8g)$ yields
\begin{align}
a+1&\ge\frac{B'+g}{8g}\ge\frac{B+g}{32g},\notag\\
b+1&\ge\frac H{4g}\ge\frac{\Lambda}{8g}.
\label{eq:interval-arms}
\end{align}
Thus $\sqrt{(a+1)(b+1)}\ge\sqrt{(B+g)\Lambda}/(16g)$. Substitution into~\eqref{eq:interval-degree} proves the claim.
\end{proof}

\subsection{From Alternating Cycles to General Types}

The next lemma gives the tree-path count bounds needed to keep the intermediate types nonnegative in the fundamental-cycle decomposition.

\begin{lemma}[Counts Along a Maximum Spanning Tree]\label{lem:tree-counts}
Let $T\ne U$ be types with the same margins, and let $h\coloneqq h(T,U)$. Assume $h<1/(64q^2)$. Let $v=v_+-v_-$ be the signed matrix of a simple alternating cycle, let $g\ge1$ be an integer, and suppose
$$
S^-=C+gv_-,\qquad S^+=C+gv_+,
$$
where $C$ is a nonnegative integer table and both $S^-$ and $S^+$ lie entrywise between $T$ and $U$. Give each edge $e$ of the complete bipartite graph on the $q$ row and $q$ column symbols weight $C_e$, and let $\tau$ be a spanning tree of maximum total weight. For each edge $e$ of the cycle outside $\tau$, let $P_e$ be the tree path joining its endpoints and define
\begin{equation}\label{eq:B-lambda}
\lambda_e\coloneqq\min_{e'\in P_e}C_{e'}.
\end{equation}
Then, for every such edge $e$,
$$
C_e\le\lambda_e,\qquad (C_e+g)\lambda_e\ge\frac{g^2}{16q^2h^2},\qquad \lambda_e\ge8qg.
$$
\end{lemma}
\begin{proof}
A tree cannot contain the entire cycle. Fix any cycle edge $e\notin\tau$, choose $e_0\in P_e$ with $C_{e_0}=\lambda_e$, and let $E$ be the cut between the two components of $\tau-\{e_0\}$. The endpoints of $e$ lie in different components, so $e\in E$. Labelling the components $0,1$ expresses $E$ as in~\eqref{eq:cuts}; hence $E\in\calE$. For every $e'\in E$, replacing $e_0$ by $e'$ gives a spanning tree. Maximality of $\tau$ therefore gives $C_{e'}\le\lambda_e$, in particular $C_e\le\lambda_e$.

Define
$$
\begin{aligned}
D_E&\coloneqq\sum_{e'\in E}(T_{e'}+U_{e'}),\qquad
V_E\coloneqq\sum_{e'\in E}|T_{e'}-U_{e'}|,\\
N_E&\coloneqq\sum_{e'\in E}
(\sqrt{T_{e'}}-\sqrt{U_{e'}})^2.
\end{aligned}
$$
The entrywise hypothesis implies $\min\{T_{e'},U_{e'}\}\le\min\{S^-_{e'},S^+_{e'}\}=C_{e'}$. Using $t+u=2\min\{t,u\}+|t-u|$ and $|E|\le q^2$, we obtain
\begin{equation}\label{eq:cut-mass-first}
\begin{aligned}
D_E&=2\sum_{e'\in E}\min\{T_{e'},U_{e'}\}+V_E\\
&\le2\sum_{e'\in E}C_{e'}+V_E\le2q^2\lambda_e+V_E.
\end{aligned}
\end{equation}
Definition~\ref{def:separation} gives $N_E\le h^2D_E$. Cauchy--Schwarz and $(\sqrt t+\sqrt u)^2\le2(t+u)$ give
\begin{align}
V_E
&=\sum_{e'\in E}|\sqrt{T_{e'}}-\sqrt{U_{e'}}|
(\sqrt{T_{e'}}+\sqrt{U_{e'}})\notag\\
&\le\sqrt{N_E}\sqrt{2D_E}\le\sqrt2\,hD_E.
\label{eq:variation-bound}
\end{align}
These inequalities also hold when $D_E=0$. Substituting into~\eqref{eq:cut-mass-first} and using $1-\sqrt2h\ge1/2$ yields
\begin{equation}\label{eq:cut-mass}
D_E\le4q^2\lambda_e.
\end{equation}

The endpoint counts at $e$ are $C_e,C_e+g$, both between $T_e,U_e$. Monotonicity of the square root and rationalization therefore give
\begin{align}
N_E&\ge(\sqrt{T_e}-\sqrt{U_e})^2 \ge(\sqrt{C_e+g}-\sqrt{C_e})^2\notag\\
&=\frac{g^2}{(\sqrt{C_e+g}+\sqrt{C_e})^2} \ge\frac{g^2}{4(C_e+g)}.
\label{eq:anchor-difference}
\end{align}
Combining this with $N_E\le h^2D_E\le4q^2h^2\lambda_e$ proves
\begin{equation}\label{eq:capacity-product}
\boxed{(C_e+g)\lambda_e\ge\frac{g^2}{16q^2h^2}.}
\end{equation}
If $\lambda_e<g$, then $C_e\le\lambda_e$ gives $(C_e+g)\lambda_e<2g^2$, whereas~\eqref{eq:capacity-product} and $h<1/(64q^2)$ give $(C_e+g)\lambda_e>256q^2g^2$. Thus $\lambda_e\ge g$, so $C_e+g\le2\lambda_e$ and $2\lambda_e^2\ge g^2/(16q^2h^2)$. Consequently,
\begin{equation}\label{eq:large-lambda}
\lambda_e\ge\frac{g}{4\sqrt2\,qh}\ge\frac{g}{8qh}\ge8qg.
\end{equation}
Since $e$ was arbitrary, all three conclusions hold for every non-tree edge of the cycle.
\end{proof}

We now combine the spanning-tree count estimates with the single-cycle lower bound to handle arbitrary pairs of types with the same margins. The target is Theorem~\ref{thm:inverse}, restated below.

\begin{inverserestatement}[Quantum Lower Bound for Two Types, Restated]
For each $q\ge2$, there is $a_q>0$ such that the following holds. If $T\ne U$ are types of total $n$ with the same margins, and a $c$-qubit protocol with arbitrary prior entanglement distinguishes every input of type $T$ from every input of type $U$ with error at most $1/3$, then
$$
c+1\ge\frac{a_q}{h(T,U)}.
$$
\end{inverserestatement}

\begin{proof}[Proof of Theorem~\ref{thm:inverse}]
Fix a protocol as in the theorem.

\paragraph{Conformal Cycles and the First Acceptance Gap.}
Treat the output identifying $U$ as acceptance. Applying a common uniform coordinate permutation before the protocol makes its acceptance probability a function $p(S)$ of the type, at no extra communication in the entanglement-assisted model. Stop any execution that would exceed $c$ qubits and assign a fixed output. This leaves promised inputs unchanged and defines $p(S)\in[0,1]$ on every type, with
$$
p(T)\le\frac13,\qquad p(U)\ge\frac23.
$$

The difference $D\coloneqq U-T$ has zero row and column sums. Orient each positive entry $D_{ab}$ from $R_a$ to $K_b$, and each negative entry in the reverse direction, with flow $|D_{ab}|$; omit zero entries. The zero sums give flow conservation at every vertex, so this is an integer circulation. If an edge remains, follow positive-flow edges until a vertex first repeats. Conservation ensures that the walk can continue, and the repeated segment is a simple directed cycle of length at most $2q$.

Subtract the minimum flow $g_j$ on this cycle from each of its edges. This preserves nonnegative integer flows and conservation, deletes at least one edge, and changes no edge direction. Record the cycle by $v_j$, with signs $+1$ on row-to-column edges and $-1$ on column-to-row edges. Its signs alternate, so its row and column sums are zero. Starting from at most $q^2$ edges, the process therefore terminates with
\begin{equation}\label{eq:conformal-decomposition}
U-T=\sum_{j=1}^m g_jv_j, \qquad m\le q^2,\quad g_j\in\mathbb Z_{>0}.
\end{equation}
The decomposition is \emph{conformal}: each nonzero $(v_j)_{ab}$ has the sign of $U_{ab}-T_{ab}$, since directions never change.

Set $S^0\coloneqq T$ and $S^j\coloneqq T+\sum_{k=1}^jg_kv_k$. Then $S^m=U$, and conformality gives
$$
\min\{T_{ab},U_{ab}\}\le(S^j)_{ab}\le\max\{T_{ab},U_{ab}\}.
$$
Thus every $S^j$ is a nonnegative integer type with the original margins. It may lie outside the promise, but $p(S^j)$ is defined. Since $S^0=T$ and $S^m=U$, the telescoping identity and the triangle inequality give
$$
\frac13\le |p(U)-p(T)| =\left|\sum_{j=1}^m \bigl(p(S^j)-p(S^{j-1})\bigr)\right| \le\sum_{j=1}^m|p(S^j)-p(S^{j-1})|.
$$
Some consecutive pair $S^-,S^+$ therefore has gap at least $1/(3q^2)$. Denote its move by $gv$, let $v_+,v_-$ be the positive and negative $0/1$ parts of $v$, and let $C$ be the entrywise minimum of the pair.
Then
\begin{equation}\label{eq:selected-cycle}
S^-=C+gv_-,\qquad S^+=C+gv_+,\qquad |p(S^+)-p(S^-)|\ge\frac1{3q^2}.
\end{equation}

\paragraph{Count Bounds from a Maximum Spanning Tree.}
Set $h\coloneqq h(T,U)>0$ and $h_0\coloneqq1/(64q^2)$. If $h\ge h_0$, choosing $a_q\le h_0$ gives $a_q/h\le1\le c+1$. Hence assume $h<h_0$. The entrywise bounds above allow us to apply Lemma~\ref{lem:tree-counts} to the selected pair. For a maximum spanning tree $\tau$ weighted by $C$, every non-tree cycle edge has tree path $P_e$ and minimum $\lambda_e$ satisfying
$$
C_e\le\lambda_e,\qquad (C_e+g)\lambda_e\ge\frac{g^2}{16q^2h^2},\qquad \lambda_e\ge8qg.
$$

\paragraph{Fundamental Cycles and the Second Acceptance Gap.}
For each $e\notin\tau$ with $v_e\ne0$, let $F_e$ be the alternating signed matrix of the fundamental cycle $e\cup P_e$, with coefficient $+1$ at $e$. Its row and column sums are zero, and
\begin{equation}\label{eq:fundamental-decomposition}
v=\sum_{\substack{e\notin\tau\\v_e\ne0}}v_eF_e.
\end{equation}
Indeed, $F_e$ has no other non-tree edge, so the two sides agree off $\tau$. Their difference has zero margins and is supported on a tree. At a leaf its sole incident coefficient must be zero; removing that edge preserves the zero margins. Repeating on the remaining forest shows that the difference vanishes.

List these non-tree edges as $e_1,\ldots,e_r$, where $1\le r\le2q$, and set
$$
W^0\coloneqq S^-,\qquad
W^j\coloneqq S^-+g\sum_{i=1}^jv_{e_i}F_{e_i}
\quad(1\le j\le r).
$$
Equation~\eqref{eq:fundamental-decomposition} gives $W^r=S^+$. Each non-tree edge changes only in its own move, from its value in $S^-$ to its value in $S^+$, so remains nonnegative. Each tree edge loses at most $rg\le2qg$ at any stage. If $e'\in P_{e_i}$, then $C_{e'}\ge\lambda_{e_i}\ge8qg$, and hence
$$
(W^j)_{e'}\ge C_{e'}-2qg \ge\lambda_{e_i}-2qg\ge\frac34\lambda_{e_i}>0.
$$
Unused tree edges do not change. Every $W^j$ is therefore a legal integer type, and all margins are preserved by the fundamental cycles.

Telescoping again gives
$$
\frac1{3q^2}\le|p(S^+)-p(S^-)| \le\sum_{j=1}^r|p(W^j)-p(W^{j-1})|.
$$
Choose $j_\ast$ with gap at least $1/(3q^2r)\ge1/(6q^3)$, and let $e_\ast=e_{j_\ast}$. For the pair $W^{j_\ast-1},W^{j_\ast}$, the common count at $e_\ast$ is $C_{e_\ast}$, and every common count on $P_{e_\ast}$ is at least $3\lambda_{e_\ast}/4$. Thus Lemma~\ref{lem:single-cycle} applies with
$$
B=C_{e_\ast},\qquad \Lambda=\lambda_{e_\ast}/2: \qquad B\le2\Lambda,\quad \Lambda\ge4qg\ge2g.
$$
Together with~\eqref{eq:capacity-product}, it yields
$$
c+1\ge b_q\frac{\sqrt{(C_{e_\ast}+g)\lambda_{e_\ast}/2}}g \ge\frac{b_q}{4\sqrt2\,qh}.
$$
Taking $a_q\coloneqq\min\{h_0,b_q/(4\sqrt2\,q)\}$ covers both ranges of $h$ and proves the theorem.
\end{proof}

\subsection{Combining Classical and Quantum Bounds}
\label{sec:consequences}
We now combine the fixed-margin upper and lower bounds. Exchanging the marginal histograms reduces a general input to a fixed-margin restriction, yielding both the classical simulation and the quantum characterization.

\paragraph{Classical Simulation.}

We first prove Theorem~\ref{thm:main}, restated below, by expressing the classical sampling cost in terms of the quantum lower bound.

\begin{mainrestatement}[Fixed-Alphabet Simulation, Restated]
For each $q\ge2$, there is a constant $C_q>0$ such that every partial function $f:[q]^n\times[q]^n\to\{-1,+1,*\}$ invariant under simultaneous coordinate permutations and satisfying $\Qstar(f)\ge1$ obeys
$$
\Rpub(f)\le C_q\,t\log(2+n/t), \qquad t\coloneqq\min\{n,\Qstar(f)^2\}.
$$
\end{mainrestatement}

\begin{proof}[Proof of Theorem~\ref{thm:main}]
Each party sends the first $q-1$ entries of its histogram; the last entry follows from their total $n$. Each count uses $\lceil\log(n+1)\rceil$ bits, so both margins become known at cost $2(q-1)\lceil\log(n+1)\rceil$.

Let $g=f_{\rho,\sigma}$ be the resulting restriction. If it has at most one label, no further communication is needed. Otherwise choose oppositely labelled types attaining $h_g$, which exist because the set of types is finite. Any protocol for $g$ distinguishes these types, so Theorem~\ref{thm:inverse} and $\Qstar(g)\le\Qstar(f)$ give
$$
h_g^{-2}\le a_q^{-2}(\Qstar(g)+1)^2 \le4a_q^{-2}\Qstar(f)^2.
$$
Let
$$
D_q\coloneqq\max\{1,4a_q^{-2}\},\qquad
u\coloneqq\min\{n,h_g^{-2}\},\qquad
t\coloneqq\min\{n,\Qstar(f)^2\}.
$$
Then $u\le D_qt$: use $u\le n=t$ if $\Qstar(f)^2\ge n$, and the preceding bound otherwise. Let $K_q$ denote the constant in Theorem~\ref{thm:classical}. By that theorem and Lemma~\ref{lem:phi},
$$
\Rpub(g)\le K_q\Phi_n(u) \le K_q\Phi_n(D_qt)\le K_qD_q\Phi_n(t).
$$
This bound is uniform over the possible restrictions. Since $t\ge1$, Lemma~\ref{lem:phi} also gives
$$
2(q-1)\lceil\log(n+1)\rceil \le4(q-1)\log(n+2)\le4(q-1)\Phi_n(t).
$$
Adding the histogram and restriction costs proves the claim with $C_q=4(q-1)+K_qD_q$, which depends only on $q$.
\end{proof}

\paragraph{Quantum Characterization.}

Applying the same lower bound together with the quantum protocol from Theorem~\ref{thm:quantum-upper} gives the characterization announced in Theorem~\ref{thm:quantum-characterization}, which we restate before giving the proof.

\begin{quantumrestatement}[Quantum Characterization, Restated]
For a permutation-invariant partial function $f$, let
$$
M(f)\coloneqq\max_{\rho,\sigma}h_{f_{\rho,\sigma}}^{-1},
$$
where a fixed-margin restriction with at most one promised label has $h_{f_{\rho,\sigma}}=1$. The maximum is over all pairs of nonnegative integer histograms $\rho,\sigma$ of total $n$. For constants $c_q,C_q>0$ depending only on $q$,
$$
c_qM(f)\le\Qstar(f)+1\le Q(f)+1 \le C_q\min\{n,M(f)\log n\}.
$$
\end{quantumrestatement}

\begin{proof}[Proof of Theorem~\ref{thm:quantum-characterization}]
Exchange the $q-1$ independent entries of each margin at cost $O_q(\log n)$ and apply Theorem~\ref{thm:quantum-upper} to the resulting restriction. Since $M(f)\ge1$, the margin exchange is absorbed by $M(f)\log n$. Alternatively, send a full input, obtaining the $O_q(n)$ cap. Increasing the constant to absorb the additive one
proves the upper bound in \eqref{eq:quantum-characterization}.

For every restriction containing both labels, choose oppositely labelled types attaining its separation parameter. Restrict any protocol for $f$ to these types and apply Theorem~\ref{thm:inverse}. Taking the maximum over margins gives $\Qstar(f)+1\ge c_qM(f)$. Restrictions with at most one label also satisfy this inequality after ensuring $c_q\le1$. The middle inequality follows because an unassisted quantum protocol is a permitted entanglement-assisted protocol. 
If $\Qstar(f)=0$, the observation in Section~\ref{sec:prelim} gives $Q(f)=0$. Otherwise, $\Qstar(f)+1\le2\Qstar(f)$, so combining the bounds yields $Q(f)\le C_q\Qstar(f)\log n$ after increasing $C_q$.
\end{proof}

\section{Symmetry-Preserving Encodings and Exponential Separations}\label{sec:counterexamples}

We first construct an XOR communication problem with a nearly linear randomized lower bound and logarithmic quantum communication. Two local encodings preserve this separation while fixing all local histograms, or even the isomorphism type of both local graphs. In these constructions, the information needed to compute the output is stored in the relative positions of symbols or the relative labelling of two copies of a rigid tree. The string construction has fixed local histograms, and the graph construction has a fixed local isomorphism type; in both cases, the encoded information lies in the relative arrangement of the two inputs.

We then consider a different mechanism. When local graph isomorphism types may vary, a universal encoding stores an arbitrary communication input in a connected graph with quadratic capacity in its vertex count. The resulting separation inherits its communication lower bound from the encoded problem. Its larger exponent reflects this greater storage capacity, rather than a stronger restriction on the local inputs.

The communication problems constructed in this section are partial functions with output alphabet $\{-1,+1,*\}$. For query functions and gadgets, we also use Boolean outputs $0,1$, which are relabelled when needed. All quantum upper bounds hold in the unassisted model $Q$; no prior entanglement or free shared randomness is required. Local encoding and decoding are free in the communication model; the results do not assert efficient local implementations.

\subsection{A Nearly Optimal XOR Communication Separation}

For a partial Boolean function $r_d$, let $Q_{\mathrm{qry}}(r_d)$ and $R_{\mathrm{qry}}(r_d)$ denote its bounded-error quantum and randomized query complexities. Let $g:\{0,1\}^{\ell}\times\{0,1\}^{\ell}\to\{0,1\}$ be a two-party Boolean function, called a gadget. Its block composition with $r_d$ is defined by
$$
(r_d\circ g)(x,y) \coloneqq r_d\bigl(g(x^{(1)},y^{(1)}),\ldots, g(x^{(d)},y^{(d)})\bigr),
$$
where each input is divided into $d$ blocks of length $\ell$. In particular, the inner-product gadget is
$$
\IP_{\ell}(x,y)\coloneqq\mathop{\oplus}\limits_{j=1}^{\ell}x_jy_j.
$$
We use the following two query-to-communication results.

\begin{lemma}[Quantum Simulation and Randomized Lifting,
{\cite{BCW98}}, {\cite{CFKMP19}}]
\label{lem:query-communication}
There are absolute constants $c,c'>0$ with the following property. For every partial Boolean function $r_d$ on $d\ge2$ bits, set $\ell\coloneqq\lceil c\log d\rceil$. Then every Boolean gadget $g:\{0,1\}^{\ell}\times\{0,1\}^{\ell}\to\{0,1\}$ satisfies
\begin{equation}\label{eq:counter-1}
Q(r_d\circ g) \le O\bigl(Q_{\mathrm{qry}}(r_d)(\ell+\log d)\bigr),
\qquad \Rpub(r_d\circ\IP_{\ell}) \ge c'\ell R_{\mathrm{qry}}(r_d).
\end{equation}
The quantum simulation requires neither prior entanglement nor shared randomness.
\end{lemma}
For the randomized statement, the cited lifting theorem increases error by at most $1/10$. Its constant-error normalization first reduces the communication protocol's error to $1/10$ by a constant number of independent repetitions. The resulting decision tree has error at most $1/5<1/3$ and query cost $O(\Rpub(r_d\circ\IP_\ell)/\ell)$. A partial function is treated as a search problem that allows either Boolean output outside its promise, so no promise verification is needed.
These formulations, including partial outer functions, are also recorded in~\cite{SSW}; the general low-discrepancy lifting theorem appears in~\cite{CFKMP21}.

The query separations used below were obtained independently by Bansal--Sinha~\cite{BS} and Sherstov--Storozhenko--Wu~\cite{SSW}, using $k$-Forrelation and $k$-Rorrelation, respectively. We use the former for its explicit construction and dependence on $k$.

\begin{lemma}[Forrelation Query Separation, {\cite{BS}}]
\label{lem:forrelation-query}
For every integer $k\ge2$ and every power of two $N\ge2$, there is an explicit partial Boolean function $r_d$ on $d\coloneqq kN$ bits with
\begin{equation}\label{eq:counter-7}
Q_{\mathrm{qry}}(r_d)=O(k2^{10k}), \qquad R_{\mathrm{qry}}(r_d)
=\Omega\left(k^{-29} \left(\frac{N}{\log(kN)}\right)^{1-1/k}\right).
\end{equation}
The implied constants are absolute.
\end{lemma}
This is the $(2^{-5k},k)$-Forrelation function of~\cite{BS}, with sign inputs identified with bits by $w\mapsto(-1)^w$.

The strong communication separation obtained directly by inner-product lifting is not, in general, an XOR function.  The following extension supplies the XOR structure needed by both encodings below, while preserving the separation. Fix $\ell\coloneqq\lceil c\log d\rceil$ as in Lemma~\ref{lem:query-communication}, and set $s\coloneqq2d\ell$. Let $z\in\{0,1\}^{s}$ be $(z^0,z^1)$, with each half divided into $d$ blocks of length $\ell$, and define
\begin{equation}\label{eq:counter-2}
\psi(z^0,z^1) \coloneqq r_d\left(\left(\mathop{\oplus}\limits_{j=1}^{\ell}z^0_{ij}z^1_{ij} \right)_{i=1}^{d}\right),
\qquad F^{\oplus}_s(A,B)\coloneqq\psi(A\oplus B).
\end{equation}
The value is undefined precisely when the argument of $r_d$ is outside its promise.  Relabelling the two Boolean outputs as $-1,+1$ does not change any complexity.

\begin{theorem}[Nearly Optimal XOR Separation]\label{thm:forrelation}
For every fixed $\varepsilon\in(0,1)$, there are infinitely many even $s$ and explicit partial functions $\psi:\{0,1\}^{s}\to\{-1,+1,*\}$ such that
\begin{equation}\label{eq:counter-3}
Q(F^{\oplus}_s)=O_{\varepsilon}(\log s), \qquad \Rpub(F^{\oplus}_s)=\Omega_{\varepsilon}(s^{1-\varepsilon}).
\end{equation}
\end{theorem}

\begin{proof}
For the classical lower bound, restrict Alice's input to $A=(x,0)$ and Bob's input to $B=(0,y)$.  On this restriction,
$$
F^{\oplus}_s((x,0),(0,y)) =(r_d\circ\IP_{\ell})(x,y).
$$
Consequently, Lemma~\ref{lem:query-communication} gives $\Rpub(F^{\oplus}_s)\ge c'\ell R_{\mathrm{qry}}(r_d)$.

For the quantum upper bound, we must handle the entire promise in \eqref{eq:counter-2}, not only this restriction.  If $A=(A^0,A^1)$ and $B=(B^0,B^1)$, the $i$th bit seen by the outer query algorithm is
\begin{equation}\label{eq:counter-5}
u_i\coloneqq\mathop{\oplus}\limits_{j=1}^{\ell} (A^0_{ij}\oplus B^0_{ij})(A^1_{ij}\oplus B^1_{ij}).
\end{equation}
Alice runs the quantum query algorithm for $r_d$ on her local registers, simulating each query to $u_i$ through communication with Bob. To simulate a query, she reversibly appends her two blocks indexed by the address register. She sends both blocks to Bob together with the address register and the query answer register. Bob reversibly computes $u_i$, XORs it into the answer register, uncomputes his workspace, and returns the registers.  Alice then uncomputes the appended blocks.  This implements the exact clean query unitary, even on a superposition of addresses, using $O(\ell+\log d)$ communicated qubits.  For nonconstant $r_d$, Alice sends the one-bit answer to Bob after the final measurement; this additional bit is absorbed by the query cost.  A constant $r_d$ requires no communication.  Hence
\begin{equation}\label{eq:counter-6}
Q(F^{\oplus}_s) \le O\bigl(Q_{\mathrm{qry}}(r_d)(\ell+\log d)\bigr).
\end{equation}
The protocol requires no shared resources.

Now take $r_d$ from Lemma~\ref{lem:forrelation-query}. Combining~\eqref{eq:counter-7} with the preceding communication bounds and $s=2d\ell$ yields the estimates
\begin{equation}\label{eq:counter-8}
Q(F^{\oplus}_s)=O(k2^{10k}\log s), \qquad \Rpub(F^{\oplus}_s) =\Omega\left(\frac{s^{1-1/k}}{k^{30}\log s}\right).
\end{equation}
For the classical estimate, the parameter substitution can be written explicitly as
$$
\ell k^{-29}\left(\frac{N}{\log(kN)}\right)^{1-1/k} =\frac{s^{1-1/k}\ell^{1/k}} {2^{1-1/k}k^{30-1/k}(\log d)^{1-1/k}} \ge\frac{s^{1-1/k}}{2k^{30}\log s},
$$
using $d=kN$, $s=2d\ell$, $\ell\ge1$, and $1\le\log d\le\log s$. For fixed $\varepsilon$, choose a fixed integer $k>1/\varepsilon$; the remaining power of $s$ absorbs the logarithm in \eqref{eq:counter-8}, proving \eqref{eq:counter-3}.
\end{proof}

\subsection{Permutation Encoding with Fixed Local Histograms}\label{sec:counter-permutations}

We now encode the XOR separation of Theorem~\ref{thm:forrelation} into pairs of permutations while keeping both local histograms fixed. This proves part~(i) of Theorem~\ref{thm:separations}, which we restate before giving the construction.

\begin{separationsrestatement}[Part (i), Restated]
For every fixed $\varepsilon\in(0,1)$, there are infinitely many $n$ and permutation-invariant partial functions $f_n:[n]^n\times[n]^n\to\{-1,+1,*\}$ such that each local input contains every symbol exactly once and
$$
Q(f_n)=O_\varepsilon(\log n),\qquad \Rpub(f_n)=\Omega_\varepsilon(n^{1-\varepsilon}).
$$
\end{separationsrestatement}

\begin{proof}
Let $F_s^\oplus(A,B)=\psi(A\oplus B)$ be a function from Theorem~\ref{thm:forrelation}. Set $n\coloneqq2s$ and identify the alphabet with
$$
\{(i,0),(i,1):i\in[s]\}.
$$
Each local input must be a permutation of all these symbols.  For every $i$, require that the set of two positions occupied by $(i,0),(i,1)$ is the same for Alice and Bob.  Each party reads the second component of the symbol at the smaller position of this pair, obtaining $A_i$ or $B_i$.  Define
\begin{equation}\label{eq:counter-9}
f_n(x,y)\coloneqq F_s^\oplus(A,B)=\psi(A\oplus B),
\end{equation}
and leave the value undefined if the local permutation condition, the common pair-position condition, or the promise of $\psi$ fails.

A simultaneous coordinate permutation maps each common pair of positions to another common pair.  It may reverse their order, but then it complements both decoded bits.  Thus $(A,B)$ changes to $(A\oplus z,B\oplus z)$ for one common mask $z$, leaving $A\oplus B$ unchanged.  The structural conditions are preserved as well, so the function and its promise are invariant under simultaneous coordinate permutations.

On promised inputs each party decodes its bit string locally and runs a protocol for $F^{\oplus}_s$.  No promise check is needed.  In the reverse direction, encode a bit $A_i$ as
$$
((i,A_i),(i,1-A_i)) \quad\text{at positions }(2i-1,2i),
$$
and encode Bob's input identically.  The encoded strings obey the structural promise and decode to the original inputs.  Both reductions are local and preserve the error probability and each resource model, so
\begin{equation}\label{eq:counter-10}
\mathsf C(f_n)=\mathsf C(F^{\oplus}_s),
\qquad \mathsf C\in\{\Rpub,Q,\Qstar\}.
\end{equation}
Theorem~\ref{thm:forrelation} and $n=2s$ give the bounds in Theorem~\ref{thm:separations}(i).

Every symbol appears exactly once on each side.  In the joint type, the $2\times2$ block for the two symbols of pair $i$ is either the identity matrix or the exchange matrix according as $A_i\oplus B_i$ is zero or one. The encoded information therefore resides entirely in the joint type, without changing either marginal histogram.
\end{proof}

\subsection{Graph Encoding with a Single Rigid Tree}\label{sec:counter-tree}

A graph communication input is an ordered pair $(A,C)$ of simple labelled graphs on $[v]$.  For $\pi\in S_v$, the graph $\pi A$ is obtained by applying $\pi$ to both endpoints of every edge.  Simultaneous relabelling invariance means $g(\pi A,\pi C)=g(A,C)$, including equality of undefined values. A graph is \emph{rigid} if its only automorphism is the identity, where an automorphism is a permutation of the vertices that preserves adjacency. A tree with this property is called a \emph{rigid tree}.

We next encode the same XOR separation into labelled copies of a single rigid tree. This proves part~(ii) of Theorem~\ref{thm:separations}, restated below.

\begin{separationsrestatement}[Part (ii), Restated]
For every fixed $\varepsilon\in(0,1)$, there are infinitely many $v$ and partial graph functions $g_v$, invariant under simultaneous vertex relabelling, such that
$$
Q(g_v)=O_\varepsilon(\log v),\qquad \Rpub(g_v)=\Omega_\varepsilon(v^{1-\varepsilon}).
$$
Each promised local graph is a labelled copy of one fixed rigid tree on $v$ vertices of maximum degree three.
\end{separationsrestatement}

\begin{proof}
Let $F_s^\oplus$ be a function from Theorem~\ref{thm:forrelation}, with $s\ge4$, and set $v\coloneqq2s$. Fix a labelled tree $B$ formed by attaching three paths of edge lengths $1,2,v-4$ to a common endpoint.  The tree has $v$ vertices and maximum degree three.  Its common endpoint is the unique vertex of degree three, and the three paths have distinct lengths at the sizes under consideration. Every automorphism therefore fixes the common endpoint and preserves each of the three paths. Each vertex is then fixed because it is uniquely determined by its path and its distance from the common endpoint. Thus $B$ is rigid, and every labelled copy has a unique representation $\sigma B$, with $\sigma\in S_v$.

Consider the subgroup
$$
H\coloneqq\langle(1\ 2),(3\ 4),\ldots,(2s-1\ 2s)\rangle,
\qquad
h_z\coloneqq\prod_{i=1}^{s}(2i-1\ 2i)^{z_i}.
$$
Its elements are indexed by $z\in\{0,1\}^{s}$, with $h_{A'}h_{B'}=h_{A'\oplus B'}$ and $h_{A'}^{-1}=h_{A'}$ for all $A',B'\in\{0,1\}^{s}$.  For local graphs $A=\sigma B$ and $C=\tau B$, define
\begin{equation}\label{eq:counter-11}
g_v(A,C)\coloneqq\psi(z)
\quad\text{if }\sigma^{-1}\tau=h_z\in H,
\end{equation}
leaving all other inputs undefined.  Values for which $\psi(z)=*$ are also undefined.  The unique representations make this definition unambiguous. Common relabelling preserves the relative permutation because $(\pi\sigma)^{-1}(\pi\tau)=\sigma^{-1}\tau$; hence it preserves the value and the promise.

For the forward reduction, Alice maps her input $A'$ to $h_{A'}B$ and Bob maps $B'$ to $h_{B'}B$.  The relative permutation is $h_{A'\oplus B'}$, so the output is $F^{\oplus}_s(A',B')$. For the reverse reduction, the promise implies $\sigma H=\tau H$, or equivalently
$$
\{\sigma(2i-1),\sigma(2i)\} =\{\tau(2i-1),\tau(2i)\} \quad(i\in[s]).
$$
Each party recovers its unique permutation using unrestricted local computation.  Sorting the two images in every pair in increasing order defines a canonical representative $r$ of the common left coset.  Both parties obtain the same $r$ locally and let $\sigma=rh_{A'}$, $\tau=rh_{B'}$.  They then run a protocol for $F_s^\oplus$ on the decoded strings, since $\sigma^{-1}\tau=h_{A'\oplus B'}$.  Thus
\begin{equation}\label{eq:counter-12}
\mathsf C(g_v)=\mathsf C(F^{\oplus}_s), \qquad \mathsf C\in\{\Rpub,Q,\Qstar\}.
\end{equation}
Theorem~\ref{thm:forrelation} gives the desired bounds. Every local graph is a connected tree of maximum degree three with the same isomorphism type.
\end{proof}

There are only $n!$ local permutation inputs in Section~\ref{sec:counter-permutations} and $v!$ labelled copies of the rigid tree here.  Sending an enumeration index gives upper bounds $O(n\log n)$ and $O(v\log v)$, respectively. For every fixed $\varepsilon>0$, our randomized lower bounds are $\Omega_\varepsilon(n^{1-\varepsilon})$ and $\Omega_\varepsilon(v^{1-\varepsilon})$. Thus their polynomial exponents approach the largest possible value under these promises. This does not determine the optimal logarithmic factors.

\subsection{General Graph Encoding with Quadratic Capacity}

If the local isomorphism type may depend on the input, a graph on $v$ vertices can store $\Theta(v^2)$ locally recoverable bits.  The following explicit encoding also makes the graph function invariant under \emph{independent} relabelling of the two local graphs, a stronger condition than simultaneous relabelling. The result is a universal encoding theorem: it preserves every communication problem, and applying it to Theorem~\ref{thm:forrelation} gives part~(iii) of Theorem~\ref{thm:separations}. Its randomized lower bound is inherited from that communication problem. The distinction from the preceding constructions is that information can now be stored in each local isomorphism type, rather than only in the relative labelling of fixed local structures.

\begin{theorem}[Encoding Communication Problems by Graphs]\label{thm:general-graphs} 
Every partial communication function $F_m$ with $m\ge1$ input bits on each side has a complexity-preserving embedding into connected simple graphs on $v=\Theta(\sqrt m)$ vertices.  The resulting partial function $\widetilde F_v$ is invariant under independent relabelling, and
\begin{equation}\label{eq:counter-13}
\mathsf C(\widetilde F_v)=\mathsf C(F_m), \qquad \mathsf C\in\{\Rpub,Q,\Qstar\}.
\end{equation}
Consequently, for every fixed $\varepsilon\in(0,1)$ there are infinitely many $v$ and explicit graph functions satisfying
\begin{equation}\label{eq:counter-14}
Q(\widetilde F_v)=O_\varepsilon(\log v),
\qquad
\Rpub(\widetilde F_v)=\Omega_\varepsilon(v^{2-\varepsilon}).
\end{equation}
\end{theorem}

\begin{proof}
Let $r\ge3$.  For a string $x\in\{0,1\}^{\binom r2}$ indexed by pairs $1\le i<j\le r$, construct a graph $G(x)$ with a backbone path $b_1-\cdots-b_r$ and data vertices $d_1,\ldots,d_r$.  Add the edges $b_id_i$, two marker leaves adjacent to each $b_i$, and two additional marker leaves adjacent to $b_1$.  Finally, add $d_id_j$ for $i<j$ if and only if $x_{ij}=1$.  This is a connected simple graph with
\begin{equation}\label{eq:counter-15}
v=4r+2.
\end{equation}

We describe an isomorphism-invariant decoding.  The backbone vertices are exactly the vertices having at least two degree-one neighbours.  Indeed, each backbone vertex has its two marker leaves.  A data vertex has no degree-one neighbour: its backbone neighbour is not a leaf, and any adjacent data vertex is also adjacent to its own backbone vertex.  A marker leaf likewise has no degree-one neighbour.  Thus the backbone is recognisable without vertex labels.

The vertex $b_1$ has four or five degree-one neighbours, while every other backbone vertex has two or three.  This identifies $b_1$ uniquely, and the induced backbone path then orders $b_1,\ldots,b_r$.  If $d_i$ participates in a data edge, it is the unique nonleaf neighbour of $b_i$ outside the backbone.  Otherwise it may be indistinguishable from a marker leaf, but the $i$th data row is all zero, so its identity is unnecessary.  This recovers every bit $x_{ij}$ and proves
\begin{equation}\label{eq:counter-16}
G(x)\cong G(x')\quad\Longleftrightarrow\quad x=x'.
\end{equation}

Given $F_m$, choose the smallest $r\ge3$ with $\binom r2\ge m$.  For $x\in\{0,1\}^m$, let $\bar x$ be its extension to $\binom r2$ bits by zero padding.  Promise that each local graph is a labelled copy of an encoding $G(\bar x)$.  Decode each graph locally and discard the padding, then set $\widetilde F_v(G(\bar x),G(\bar y))\coloneqq F_m(x,y)$; if either structural promise or the promise of $F_m$ fails, leave the value undefined.  The decoding is invariant under each party's relabelling separately.  Forward encoding and reverse decoding are both zero-communication reductions, proving \eqref{eq:counter-13}, and minimality of $r$ gives $v=\Theta(\sqrt m)$.

Apply this embedding to Theorem~\ref{thm:forrelation} with exponent loss $\varepsilon/2$. The lower bound $\Omega_\varepsilon(m^{1-\varepsilon/2})$ becomes $\Omega_\varepsilon(v^{2-\varepsilon})$, while $\log m=\Theta(\log v)$. This proves Theorem~\ref{thm:separations}(iii).
\end{proof}

Sending a full adjacency matrix costs $O(v^2)$ bits. Thus the universal encoding uses the available local information capacity up to a constant factor, and its randomized lower bound $\Omega_\varepsilon(v^{2-\varepsilon})$ approaches the largest possible polynomial exponent as $\varepsilon\to0$. This capacity statement complements the fixed-tree construction of Section~\ref{sec:counter-tree}, whose local isomorphism type contains no varying information. In both cases, the communication separation is inherited from the underlying communication problem; the encodings identify which information survives the specified symmetries.

All of these separations concern partial functions: assigning arbitrary values outside the promise does not establish the same quantum upper bounds for a total function.

\section{Exponent Obstructions and Remaining Limitations}\label{sec:limits}
We distinguish several questions about the bounds. Set disjointness prevents a fixed decrease in the quadratic exponent of quantum communication. Hamming-weight encodings show that a simulation must depend on the input length and prevent a fixed decrease in the logarithmic exponent. A different, fixed-margin construction shows that the logarithm in the quantum upper bound in terms of $h_f^{-1}$ is necessary. We then examine the dependence on the alphabet size and explain what these examples leave unresolved about classical simulation with fixed margins. The examples do not establish optimality of~\eqref{eq:main} throughout the full range of quantum communication complexities.

\subsection{Tightness of the Quadratic Exponent}

Binary set disjointness $\DISJ_n(x,y)$ is $+1$ when the supports of $x,y\in\{0,1\}^n$ are disjoint and $-1$ otherwise. \begin{lemma}[Communication Complexity of Disjointness, {\cite{Razborov92,Razborov03,AASpatial05}}]\label{lem:disjointness-bounds} The bounded-error complexities satisfy
$$
\Rpub(\DISJ_n)=\Theta(n),\qquad \Qstar(\DISJ_n)=\Theta(\sqrt n),\qquad Q(\DISJ_n)=\Theta(\sqrt n).
$$
The randomized lower bound is from~\cite{Razborov92}; its distributional bound also applies to public-coin protocols by averaging over their randomness. The entanglement-assisted quantum lower bound is \cite{Razborov03}, and the unassisted quantum upper bound is proved in~\cite{AASpatial05}.
\end{lemma}

\begin{proposition}[Quadratic Exponent Obstruction]\label{prop:quadratic-obstruction} 

For every fixed $\eta>0$ and $K\ge0$, there is no constant $C$ such that all binary permutation-invariant partial functions $f$ with $\Qstar(f)\ge1$ satisfy
$$
\Rpub(f)\le C\Qstar(f)^{2-\eta}(\log n)^K.
$$
\end{proposition}
\begin{proof}
Set disjointness is permutation-invariant. By Lemma~\ref{lem:disjointness-bounds}, for any fixed $\eta>0$, the two-sided quantum estimate gives $\Qstar(\DISJ_n)^{2-\eta}=\Theta(n^{1-\eta/2})$, including when $2-\eta\le0$.  Thus, for fixed $K\ge0$, a universal bound $O_q(\Qstar(f)^{2-\eta}(\log n)^K)$ would imply
$$
\Rpub(\DISJ_n) =O\bigl(n^{1-\eta/2}(\log n)^K\bigr) =o(n),
$$
contradicting the randomized lower bound, since $n^{-\eta/2}(\log n)^K\to0$.  Thus no fixed polylogarithmic factor in $n$ allows a fixed positive decrease in the quadratic exponent.
\end{proof}

\subsection{Necessity of the Input-Length Dependence}
\label{sec:weight-obstruction}

Ghazi, Kamath, and Sudan~\cite{GKS16} observed the universal expressiveness of binary permutation-invariant functions under exponential input-length expansion. Combined with the communication separations of Bansal and Sinha~\cite[Corollary~1.6(a)]{BS} and Sherstov, Storozhenko, and Wu~\cite{SSW}, this observation already yields a separation of the form \eqref{eq:weight-separation}. We give the reduction explicitly using the XOR communication function in Theorem~\ref{thm:forrelation}. The following lemma records the Hamming-weight encoding and its exact preservation of the three communication models used here.

\begin{lemma}[Hamming-Weight Encoding, following {\cite{GKS16}}]\label{lem:weight-encoding}
Let $m\ge2$, let $F_m:\{0,1\}^m\times\{0,1\}^m\to\{-1,+1,*\}$ be a partial communication function, and set $n\coloneqq2^m-1$.  Let $\bin_m(t)$ denote the $m$-bit representation of $t\in\{0,\ldots,n\}$.  The function
$$
f_n(x,y)\coloneqq F_m\bigl(\bin_m(|x|),\bin_m(|y|)\bigr)
$$
on binary inputs is invariant under independent coordinate permutations, and
$$
\mathsf C(f_n)=\mathsf C(F_m), \qquad\mathsf C\in\{\Rpub,Q,\Qstar\}.
$$
\end{lemma}
\begin{proof}
Coordinate permutations preserve both weights, hence the output and promise.  Local weight computation reduces $f_n$ to $F_m$, giving $\mathsf C(f_n)\le\mathsf C(F_m)$.  Conversely, if an $m$-bit string $a$ represents the integer $t$, encode it as $1^t0^{n-t}$. The resulting string has weight $t$, so applying $\bin_m$ recovers $a$.  This local reduction gives the reverse inequality in every model.
\end{proof}

Combining Theorem~\ref{thm:forrelation} with Lemma~\ref{lem:weight-encoding} gives the following quantitative consequence.

\begin{logrestatement}[Logarithmic Exponent Obstruction, Restated]
For every fixed $p\ge0$ and $\delta\in(0,1]$, there is no constant $C$ such that all binary permutation-invariant partial functions $f$ with $\Qstar(f)\ge1$ satisfy
$$
\Rpub(f)\le C\Qstar(f)^p(\log n)^{1-\delta}.
$$
\end{logrestatement}
\begin{proof}[Proof of Theorem~\ref{thm:log-obstruction}]
Fix $\varepsilon\in(0,1)$ and take $F_m=F_m^{\oplus}$ from Theorem~\ref{thm:forrelation}, so $Q(F_m)=O_\varepsilon(\log m)$ and $\Rpub(F_m)=\Omega_\varepsilon(m^{1-\varepsilon})$. Since $n=2^m-1$, we have $m=\log(n+1)=\Theta(\log n)$ and $\log m=\Theta(\log\log n)$.  Lemma~\ref{lem:weight-encoding} gives
\begin{equation}\label{eq:weight-separation}
Q(f_n)=O_\varepsilon(\log\log n),\qquad \Rpub(f_n)=\Omega_\varepsilon((\log n)^{1-\varepsilon})
\end{equation}
along an unbounded sequence of lengths.  Fix $p\ge0$ and $\delta\in(0,1]$, and take $\varepsilon=\delta/2$. Using $\Qstar(f_n)\le Q(f_n)$ and dividing the lower bound by $O((\log\log n)^p(\log n)^{1-\delta})$ would give
$$
\frac{(\log n)^{\delta/2}}{(\log\log n)^p}=O(1).
$$
The ratio diverges as $n\to\infty$, giving a contradiction.
\end{proof}

The same construction gives the corollary stated in the introduction, which we restate before proving it.

\begin{polynomialrestatement}
For every fixed $\varepsilon\in(0,1)$, there is a family of binary partial functions, invariant under independent coordinate permutations, such that along an unbounded sequence of lengths,
$$
Q(f_n)=O_\varepsilon(\log\log n),\qquad \Rpub(f_n)=\Omega_\varepsilon((\log n)^{1-\varepsilon}).
$$
Consequently, for every fixed $p\ge0$,
$$
\frac{\Rpub(f_n)}{(Q(f_n)+1)^p}\longrightarrow\infty.
$$
The same conclusion holds with $Q$ replaced by $\Qstar$.
\end{polynomialrestatement}
\begin{proof}[Proof of Corollary~\ref{cor:no-pure-polynomial}]
Fix $\varepsilon\in(0,1)$ and use the family in~\eqref{eq:weight-separation}. For every fixed $p\ge0$,
$$
\frac{\Rpub(f_n)}{(Q(f_n)+1)^p} \ge\Omega_{\varepsilon,p}\!\left( \frac{(\log n)^{1-\varepsilon}}{(\log\log n)^p}\right) \longrightarrow\infty.
$$
The inequality $\Qstar(f_n)\le Q(f_n)$ gives the last assertion.
\end{proof}

Corollary~\ref{cor:no-pure-polynomial} shows that input-length dependence is necessary: a bound polynomial in quantum communication complexity alone is not possible, even for a binary alphabet. Thus the polynomial relation between randomized and quantum communication complexity conjectured by Guan et al.~\cite{GHYY} requires an additional dependence on the input length, even for binary permutation-invariant partial functions. The quantum upper bound requires neither prior entanglement nor shared randomness, whereas the classical lower bound allows public randomness.

Sending Alice's weight gives an $O(\log n)$ upper bound, but does not supply a matching lower bound.  Since $\varepsilon$ is fixed and cannot be set to zero in~\eqref{eq:weight-separation}, the obstruction rules out fixed decreases in the logarithmic exponent, not all smaller logarithmic factors.

These examples are constant on every fixed-margin restriction. They therefore do not rule out a purely quadratic simulation for fixed-margin restrictions.

\subsection{Necessity of the Logarithmic Factor in the Quantum Upper Bound}
The logarithmic factor in Theorem~\ref{thm:quantum-upper} is necessary over a range of separation parameters, even with binary fixed margins.
\begin{lemma}[Inner-Product Lower Bound, {\cite{CVDNT}}]
\label{lem:inner-product-lower}
For the Boolean inner-product function on $m$ bits per party, $\Qstar(\IP_m)=\Omega(m)$.
\end{lemma}

\begin{proposition}[Intersection Parity with Fixed Margins]\label{prop:membership}
Let $n$ be even and $1\le a\le n/4$ be an integer.  On the domain $|x|=a$, $|y|=n/2$, define
$$
f_{n,a}(x,y)\coloneqq|x\wedge y|\pmod2,
$$
with Boolean outputs relabelled as $\{-1,+1\}$.  This partial function is permutation-invariant and satisfies
$$
\frac1{2a}\le h_{f_{n,a}}\le\frac3a.
$$
For $n=a2^{m+1}$ with an integer $m\ge1$,
$$
\mathsf C(f_{n,a}) =\Theta\!\left(a\log\frac na\right), \qquad \mathsf C\in\{\Qstar,Q,\Rpub\}.
$$
When $a\le\sqrt n$ on these lengths, this equals $\Theta(h_{f_{n,a}}^{-1}\log n)$.  At $a=1$, we have $h_{f_{n,1}}=1$ exactly and all three complexities are $\Theta(\log n)$ on these lengths.
\end{proposition}
\begin{proof}
Both the domain and intersection parity are permutation-invariant. For $c\coloneqq|x\wedge y|$, the joint type, indexed by $0,1$, is
$$
T_c\coloneqq\begin{pmatrix} 
n/2-a+c&n/2-c\\
a-c&c
\end{pmatrix},\qquad c=0,\ldots,a.
$$
Every listed $c$ is realizable, and opposite labels correspond exactly to intersection counts of different parity. 

\paragraph{Evaluating the Type-Separation Parameter.}
For two types $T_c,T_d$ with opposite labels, we have $|c-d|\ge1$. The cut isolating row $1$ has denominator $2a$ in~\eqref{eq:h}.  Its entries are at most $a$, so each corresponding pair satisfies
$$
(\sqrt u-\sqrt v)^2 =\frac{(u-v)^2}{(\sqrt u+\sqrt v)^2} \ge\frac{(u-v)^2}{4a}.
$$
Both differences have magnitude $|c-d|$, giving
$$
h(T_c,T_d)^2\ge \frac{2(c-d)^2/(4a)}{2a} =\frac{(c-d)^2}{4a^2}.
$$
Minimizing over pairs of types with opposite labels yields $h_{f_{n,a}}\ge1/(2a)$. For the reverse bound, take $c\coloneqq\lfloor a/2\rfloor$ and $d\coloneqq c+1$. For $a\ge4$, every row-$1$ entry in these types is at least $a/4$, and every row-$0$ entry is at least $n/4\ge a$.  Each corresponding pair of entries differs by one, so the same square-root identity bounds each term in the numerator by $1/a$. The seven nonempty binary cut events are the two rows, two columns, two perfect matchings, and all four entries.  The row counts are $a,n-a$; the column counts are $n/2$; and the matching counts in $T_j$ are
$$
n/2-a+2j,\qquad n/2+a-2j.
$$
Both matching counts are at least $a$ because $n/2\ge2a$ and $0\le j\le a$.  Thus every event has count at least $a$ in each type, so its denominator is at least $2a$ and its numerator at most $4/a$. Hence $h(T_c,T_d)\le\sqrt2/a$.  For $a\le3$, use $h_{f_{n,a}}\le1\le3/a$, proving the claimed upper bound. At $a=1$, the two types are
$$
T_0=
\begin{pmatrix}
n/2-1 & n/2\\
1 & 0
\end{pmatrix},
\qquad
T_1=
\begin{pmatrix}
n/2 & n/2-1\\
0 & 1
\end{pmatrix}.
$$
The row-$1$ cut has numerator and denominator both equal to $2$. Together with $h(T_0,T_1)\le1$, this gives $h_{f_{n,1}}=1$. \paragraph{Communication Bounds for Intersection Parity.} For the lower bound, take $n=a2^{m+1}$ and index coordinates by
$$
(i,u,z)\in[a]\times\{0,1\}^m\times\{0,1\}.
$$
Split each input to $\IP_{am}$ into blocks $A_i,B_i\in\{0,1\}^m$. Alice forms $x_A$ by putting a one at $(i,A_i,0)$ for each $i$ and zeros elsewhere, giving weight $a$.  Bob sets
$$
y_B(i,u,z)\coloneqq\IP_m(u,B_i)\oplus z.
$$
For each $(i,u)$, exactly one choice of $z$ gives a one, so Bob's weight is $n/2$.  In the Boolean output convention,
$$
f_{n,a}(x_A,y_B) =\mathop{\oplus}\limits_{i=1}^a\IP_m(A_i,B_i) =\IP_{am}(A,B).
$$
This local reduction preserves the full promise. Lemma~\ref{lem:inner-product-lower} gives an $\Omega(am)$ entanglement-assisted quantum lower bound. Since $\log(n/a)=m+1\le2m$ for $m\ge1$,
$$
\Qstar(f_{n,a})\ge\Omega(am) =\Omega\!\left(a\log\frac na\right).
$$
For the upper bound, Alice sends the index of her weight-$a$ input in a fixed enumeration, using
$$
\left\lceil\log\binom na\right\rceil \le a\log(en/a)+1 =O\!\left(a\log\frac na\right).
$$
Bob computes the output exactly.  This deterministic protocol works in all three models; together with $\Qstar(f_{n,a})\le Q(f_{n,a})$ and $\Qstar(f_{n,a})\le\Rpub(f_{n,a})$, it proves their common order.  Finally, for $a\le\sqrt n$, $\tfrac12\log n\le\log(n/a)\le\log n$.  Combining this with the parameter bounds gives $\Theta(h_{f_{n,a}}^{-1}\log n)$, including $\Theta(\log n)$ at $a=1$.
\end{proof} 

For example, take $a\coloneqq2^{m+1}$ and $n\coloneqq a^2$.  Then 
$$
h_{f_{n,a}}^{-1}=\Theta(\sqrt n),\qquad \Qstar(f_{n,a})=\Theta(\sqrt n\log n).
$$
Thus no universal bound $O_q(h_f^{-1}+\log n)$ is possible: $O_q(h_f^{-1}\log n)$ is attained throughout the proposition's range $a\le\sqrt n$.  This is consistent with the bounds for AND-symmetric predicates~\cite{Suruga}, which concern total functions and use different parameters. These examples establish the necessity of the logarithm relative to $h_f^{-1}$ in the quantum upper bound. Their randomized and quantum complexities have the same order, so they do not establish a logarithmic loss in a comparison of $\Rpub(f)$ with $\Qstar(f)^2$.

\subsection{Dependence on a Growing Alphabet}\label{sec:alphabet-dependence}

Theorem~\ref{thm:main} gives a simulation for each fixed alphabet, with a coefficient $C_q$ that may depend on $q$. Proposition~\ref{prop:alphabet-obstructions} shows that suitable subpolynomial alphabets rule out uniform bounds polynomial in $Q(f),\log n,\log q$, while polylogarithmic alphabets with exponent greater than one already rule out $O(\Qstar(f)^2\log n)$ with a constant independent of $q$. It also forces $C_q\ge q^{1-o(1)}$ along an unbounded sequence of alphabet sizes; all three conclusions hold even when both local histograms are fixed.

\begin{proposition}[Growing-Alphabet Obstructions]\label{prop:alphabet-obstructions}
The following statements hold even for fixed-margin partial functions.
\begin{enumerate}[label=(\roman*)]
\item There are permutation-invariant families with $q=n^{o(1)}$ for which $\Rpub(f)$ exceeds every fixed polynomial in $Q(f),\log n,\log q$.
\item For every fixed $c>1$, there are permutation-invariant families with $q=\Theta((\log n)^c)$ for which $\Rpub(f)$ is not $O(\Qstar(f)^2\log n)$ with an absolute implied constant.
\item Any choice of coefficients $C_q$ making the bound in Theorem~\ref{thm:main} valid must satisfy $C_q\ge q^{1-o(1)}$ along an unbounded sequence of alphabet sizes.
\end{enumerate}
\end{proposition}
\begin{proof}
Fix $\varepsilon\in(0,1)$.  Start with a length-$q_0$ instance of Theorem~\ref{thm:separations}(i), with each of its $q_0$ symbols occurring once.  For any $n\ge q_0$, add $n-q_0$ copies of a new dummy symbol at common positions.  Require these positions to coincide and the remaining symbols to satisfy the original paired-position promise.  The alphabet has size $q\coloneqq q_0+1$, and both local histograms are fixed. Each party can delete the dummy positions and decode its symbol pairs. Conversely, use fixed positions for the original encoding and fill the rest with dummies.  Both reductions are local; common permutations preserve the promise and decoded XOR.  Thus, writing $Q,\Qstar,\Rpub$ for the padded function's complexities, we retain
\begin{equation}\label{eq:padded-separation}
Q=O_\varepsilon(\log q),\qquad \Rpub=\Omega_\varepsilon(q^{1-\varepsilon}).
\end{equation}
\paragraph{Superpolylogarithmic Alphabets.}
If $\log q/\log\log n\to\infty$, this lower bound exceeds every fixed polynomial in $Q,\log n,\log q$.  Indeed, the logarithm of such a polynomial is $O(\log\log q+\log\log n)$, whereas that of the lower bound is $(1-\varepsilon)\log q+O_\varepsilon(1)$.  To obtain such a family, let $q=q_0+1$ range over the unbounded sequence of alphabet sizes furnished by the construction above, and set
$$
n\coloneqq\left\lceil 2^{(\log q)^2}\right\rceil.
$$
For sufficiently large $q$, we have $n\ge q_0$ and $\log n=(\log q)^2+o(1)$. Hence
$$
\frac{\log q}{\log\log n}\longrightarrow\infty, \qquad q=\exp\!\bigl(\Theta(\sqrt{\log n})\bigr)=n^{o(1)}.
$$
This proves part~(i), with $\varepsilon$ fixed throughout.
\paragraph{Polylogarithmic Alphabets.}
For $q=\Theta((\log n)^c)$, padding does not exclude polynomial bounds of arbitrary fixed degree.  It does exclude $O((\Qstar)^2\log n)$ with an absolute constant whenever $c>1$. Choose $0<\varepsilon<1-1/c$.  Equation~\eqref{eq:padded-separation} gives
$$
Q=O_\varepsilon(\log\log n),\qquad \Rpub=\Omega_\varepsilon((\log n)^{c(1-\varepsilon)}).
$$
The proposed bound is $O_\varepsilon(\log n\,(\log\log n)^2)$.  Since $c(1-\varepsilon)>1$, the ratio of the lower bound to this quantity diverges as
$$
\frac{(\log n)^{c(1-\varepsilon)-1}}{(\log\log n)^2}\longrightarrow\infty.
$$
A compatible subsequence is obtained by taking the constructed alphabet sizes and setting $n\coloneqq\lceil\exp(q^{1/c})\rceil$.  Then $n\ge q_0$ for large sizes and $q=\Theta((\log n)^c)$.
\paragraph{Growth of the Simulation Constant.}
For each fixed $\varepsilon\in(0,1)$, the unpadded family in Theorem~\ref{thm:separations}(i) has $q=n$ and satisfies
$$
Q=O_\varepsilon(\log q),\qquad \Rpub=\Omega_\varepsilon(q^{1-\varepsilon}).
$$
For $t=\min\{q,(\Qstar)^2\}$ in~\eqref{eq:main}, the inequality $\Qstar\le Q$ gives $t=O_\varepsilon((\log q)^2)$. Since $t\ge1$,
$$
t\log(2+q/t)=O_\varepsilon((\log q)^3).
$$
Dividing the classical lower bound by this quantity yields
$$
C_q\ge\Omega_\varepsilon\!\left( \frac{q^{1-\varepsilon}}{(\log q)^3}\right)
$$
along an unbounded sequence of alphabet sizes for each fixed $\varepsilon$. For each integer $j\ge1$, take $\varepsilon=1/(2j)$ and successively choose increasing sizes $q_j$ from the corresponding sequences, large enough that $C_{q_j}\ge q_j^{1-1/j}$. Such a choice is possible because $q^{1/(2j)}/(\log q)^3\to\infty$ for fixed $j$. Thus $C_{q_j}\ge q_j^{1-o(1)}$, without requiring constants uniform in $\varepsilon$.
\end{proof}

The lower bound on $C_q$ does not identify a transition between fixed and growing alphabets. The following estimate makes explicit which part of the upper-bound dependence is controlled by our argument.

\begin{proposition}[Alphabet Dependence of the Present Proof]
\label{prop:alphabet-upper}
Let $K(q)\coloneqq\max\{1,K_{q,1/(4q)}\}$, where $K_{q,1/(4q)}$ is the approximation constant in Lemma~\ref{lem:profile}. The coefficient in Theorem~\ref{thm:main} can be chosen to satisfy
\begin{equation}\label{eq:alphabet-upper}
C_q\le 2^{O(q)}K(q)^2,
\end{equation}
where the constant in $O(q)$ is absolute.
\end{proposition}
Appendix~\ref{sec:aux-alphabet} tracks the explicit classical parameters and the losses in the quantum lower bound to prove this estimate.

\paragraph{The Remaining Analytic Dependence.}
The estimate~\eqref{eq:alphabet-upper} isolates the dependence that is not made explicit as a function of $q$ in this paper. For $\beta=1/(4q)$, the proof of Lemma~\ref{lem:operator-contraction} can use $\alpha_q=(4q)^{-(2q+1)}$ in the multislice contraction theorem. If $K_{\mathrm{ext}}\ge1$ and $\eta_{\mathrm{ext}}>0$ are its constants at this admissibility parameter, chosen uniformly over active alphabet sizes at most $q$, the $q$th-power argument gives contraction constants $K_0=K_{\mathrm{ext}}^{1/(2q)}$ and $\gamma=\ln(1+\eta_{\mathrm{ext}})/(2q)$. The truncation in Lemma~\ref{lem:profile} consequently allows
$$
K(q)=O\!\left(1+\frac{q+\ln K_{\mathrm{ext}}}{\ln(1+\eta_{\mathrm{ext}})} \right).
$$
Thus a fully explicit growth bound for $C_q$ would require tracking the quantitative dependence of these external contraction constants. We do not extract that dependence here, and make no claim that \eqref{eq:alphabet-upper} matches the necessary growth in Proposition~\ref{prop:alphabet-obstructions}.

\paragraph{Very Slowly Growing Alphabets.}
The existence of finite constants $C_q$ for every $q$ does already allow some unbounded alphabet growth. To see this, let $\widehat C_r\coloneqq\max_{2\le j\le r}\max\{1,C_j\}$ and, for all sufficiently large $n$, choose the largest integer $2\le r\le\lfloor\log n\rfloor$ with $\widehat C_r\le\log n$; call it $q(n)$. Every fixed $r$ eventually satisfies these inequalities, so $q(n)\to\infty$. For alphabets of size at most $q(n)$, Theorem~\ref{thm:main} gives
$$
\Rpub(f)=O\!\left(\Qstar(f)^2(\log n)^2\right),
$$
with an absolute implied constant. This argument supplies no explicit rate of growth for $q(n)$. In particular, it does not determine which specified polylogarithmic alphabet sizes admit a uniform simulation polynomial in quantum communication and $\log n$.

\section{Conclusion}\label{sec:conclusion}

For permutation-invariant partial functions over a fixed alphabet, we established a classical simulation with a quadratic dependence on entanglement-assisted quantum communication complexity and a logarithmic overhead in the input length. The refined bound is $\Rpub(f)\le C_q t\log(2+n/t)$, where $t=\min\{n,\Qstar(f)^2\}$. We also characterized quantum communication complexity by a combinatorial separation parameter up to one logarithmic factor. The corresponding quantum protocol uses neither prior entanglement nor shared randomness. These results extend and sharpen the binary-alphabet bounds of Guan et al.~\cite{GHYY}. The proofs combine batch recovery and coherent count estimation with a reduction from multidimensional joint types to one-variable polynomial approximation.

The dependence on input length remains necessary even over a binary alphabet. The encoding observation of~\cite{GKS16}, applied to known communication separations, yields functions with quantum communication complexity $O_\varepsilon(\log\log n)$ and randomized communication complexity $\Omega_\varepsilon((\log n)^{1-\varepsilon})$, for every fixed $\varepsilon\in(0,1)$. Thus no fixed polynomial in quantum communication complexity alone bounds randomized communication complexity, and the logarithmic exponent in the simulation cannot be decreased by a fixed positive amount.

Growing alphabets and graph symmetries permit exponential quantum advantages. For every fixed $\varepsilon\in(0,1)$, our permutation-invariant functions over an $n$-symbol alphabet have quantum communication complexity $O_\varepsilon(\log n)$ and randomized communication complexity $\Omega_\varepsilon(n^{1-\varepsilon})$, with every local histogram fixed. A graph encoding preserves such a separation even when both local inputs are copies of a single rigid tree. Allowing the graph isomorphism types to vary gives connected, independently relabelling-invariant graph problems with quantum communication complexity $O_\varepsilon(\log v)$ and randomized communication complexity $\Omega_\varepsilon(v^{2-\varepsilon})$. 

\section*{Acknowledgments}
\addcontentsline{toc}{section}{Acknowledgments}
We thank Boji Xu, Penghui Yao, and Jiadong Zhu for helpful discussions. Yunqi Huang was supported by the Guangdong Provincial Quantum Science Strategic Initiative (Grants No.\ GDZX2503001). Zekun Ye was supported by start-up research funding from Fuzhou University (Grant No.\ XRC-26004). The authors used ChatGPT (OpenAI) to assist with developing and refining constructions and proof arguments, checking derivations, and drafting and revising the manuscript. The authors take full responsibility for the correctness of all results and proofs, the attribution of prior work, and the final content of the paper.

\clearpage
\appendix
\addtocontents{toc}{\protect\newpage}

\section{Quantum Upper Bound}\label{sec:quantum-upper}

This appendix proves Theorem~\ref{thm:quantum-upper}, the quantum analogue of the classical count-estimation protocol. The construction uses neither prior entanglement nor free shared randomness and is independent of the lower-bound argument in Section~\ref{sec:quantum}.

\subsection{Randomness Reduction and Amplitude Estimation}

We use two standard estimates, stated below with their original sources. Hoeffding's inequality supplies the concentration bound for a finite random-tape table; amplitude estimation then estimates the table's output probabilities coherently.

\begin{lemma}[Hoeffding's Inequality {\cite{Hoeffding63}}]
\label{thm:hoeffding}
Let $X_1,\ldots,X_K$ be independent random variables taking values in $[0,1]$. For every $\eta>0$,
$$
\Prb\!\left[\left|\frac1K\sum_{i=1}^K(X_i-\E X_i)\right|>\eta\right] \le2\exp(-2K\eta^2).
$$
\end{lemma}

\begin{lemma}[Finite Random-Tape Tables, Adapted from {\cite{Newman91}}]\label{lem:finite-tape}
Consider a randomized experiment on input pairs $(x,y)\in[q]^n\times[q]^n$, with output in a finite set $\mathcal Z$. For every $0<\eta<1$, there is a fixed list of $K$ random tapes, where $K$ is a power of two and
$$
K=O\!\left(\eta^{-2}(n\log q+\log|\mathcal Z|)\right),
$$
such that choosing a uniform tape from this list changes the probability of any output by at most $\eta$, simultaneously for all input pairs. Any event that is impossible for every original tape remains impossible on this list.
\end{lemma}
\begin{proof}
Sample $K$ independent tapes. For each input pair and output $z$, Hoeffding's inequality bounds the probability that its empirical frequency differs from its original probability by more than $\eta$ by $2e^{-2K\eta^2}$. A union bound over all $q^{2n}|\mathcal Z|$ choices therefore bounds the probability of any violation by
$$
2q^{2n}|\mathcal Z|e^{-2K\eta^2}<1
$$
for sufficiently large $K$ of the stated order. Rounding up to a power of two costs at most a factor of two. Thus a suitable list exists; impossible events stay impossible because all entries are original tapes.
\end{proof}

\begin{lemma}[Amplitude Estimation {\cite{BHMT}}]
\label{thm:amplitude-estimation}
Let $U$ be a state-preparation unitary and let $\Pi$ be the orthogonal projection onto a specified outcome. Let $a\coloneqq\|\Pi U|0\rangle\|^2$ be the probability of that outcome. Assume access to $U$, $U^{-1}$, the reflections $I-2|0\rangle\langle0|$ and $I-2\Pi$, and their controlled versions. For every integer $t\ge1$, amplitude estimation uses $O(t)$ calls to these operations and produces an estimate $\widehat a\in[0,1]$ such that, with probability at least $8/\pi^2$,
\begin{equation}
|\widehat a-a|\le \frac{2\pi\sqrt{a(1-a)}}t+\frac{\pi^2}{t^2}.
\label{eq:qupper-28}
\end{equation}
\end{lemma}

\begin{lemma}[Amplitude Estimation in Square-Root Distance]
\label{lem:sqrt-amplitude}
Under the assumptions of Lemma~\ref{thm:amplitude-estimation}, the same estimate satisfies
$$
|\sqrt{\widehat a}-\sqrt a|\le C/t
$$
with probability at least $8/\pi^2$, for a universal constant $C$. If each call to $U$, its inverse, or their controlled versions costs $c$ qubits and the reflections are local, the total communication is $O(tc)$.
\end{lemma}
\begin{proof}
Condition on the event in~\eqref{eq:qupper-28}. If $a\le t^{-2}$, that inequality gives $\widehat a\le(1+2\pi+\pi^2)t^{-2}$, so the square-root difference is $O(t^{-1})$. If $a>t^{-2}$, rationalizing the difference gives
$$
|\sqrt{\widehat a}-\sqrt a| =\frac{|\widehat a-a|}{\sqrt{\widehat a}+\sqrt a} \le\frac{2\pi}{t}+\frac{\pi^2}{t^2\sqrt a} \le\frac{2\pi+\pi^2}{t}.
$$
The communication bound follows by counting the calls.
\end{proof}

\subsection{Quantum Protocol and Communication Cost}

\begin{quantumupperrestatement}[Quantum Upper Bound, Restated]
For every fixed $q\ge2$, there is a constant $C_q>0$ such that every fixed-margin restriction $g\coloneqq f_{\rho,\sigma}$ of a permutation-invariant partial function $f:[q]^n\times[q]^n\to\{-1,+1,*\}$ satisfies
$$
Q(g)\le C_q\min\{n,h_g^{-1}\log n\}.
$$
\end{quantumupperrestatement}

\begin{proof}
If the promise has at most one output label, no communication is needed. Otherwise let $h\coloneqq h_g\in(0,1]$, and use the cut family $\mathcal E$ and auxiliary distributions $J_W$ from Section~\ref{sec:prelim}. Fix binary encodings $(\alpha,\beta)$ for each cut as in~\eqref{eq:cuts}. Recall that $|\mathcal E|\le2^{2q-1}-1$ and that their Hellinger distance is given by~\eqref{eq:hellinger}. The distributions are used only for analysis; the protocol below samples finitely many random bits.

\paragraph{Precision Required by Integer Joint Types.}
Two distinct integer types $T,U$ have an entry differing by at least one. Since all entries lie between $0$ and $n$, that entry satisfies
$$
(\sqrt{T_e}-\sqrt{U_e})^2 =\frac{(T_e-U_e)^2}{(\sqrt{T_e}+\sqrt{U_e})^2} \ge\frac1{4n}.
$$
The event containing all ordered symbol pairs belongs to $\mathcal E$ and has denominator $2n$. Consequently,
\begin{equation}
h\ge\frac1{\sqrt8\,n}, \qquad \log(1/h)=O(\log n).
\label{eq:qupper-22}
\end{equation}

\paragraph{Estimating the Disagreement Counts.}
Let $T$ be the true type. Apply Lemma~\ref{lem:scale} to each $E\in\mathcal E$, with failure probability at most $1/(40|\mathcal E|)$, to estimate $k\coloneqq T(E)$. By the union bound, with total failure probability at most $1/40$, its estimates satisfy
$$
k=0\ \Longrightarrow\ \widetilde k=0, \qquad k>0\ \Longrightarrow\ k/\kappa\le\widetilde k\le\kappa k,
$$
for some $\kappa\ge1$ depending only on $q$. The total communication is $O_q(\log n)$. The zero case is one-sided: a zero count is never reported as positive, for any choice of the public random tape.

This preprocessing can be implemented without shared randomness. Apply Lemma~\ref{lem:finite-tape} to its failure indicator, uniformly over all input pairs. For constant additional error $1/40$, a table of $O_q(n)$ full random tapes suffices; its size may be chosen to be a power of two. Alice chooses a uniform table index using private randomness and sends that index to Bob at cost $O_q(\log n)$. The zero-case guarantee survives the restriction to any table. Thus the probability that some estimate $\widetilde k$ violates these bounds is at most $1/20$.

If $\widetilde k=0$, set the estimated count vector $\widehat W^E\coloneqq0$. Otherwise clamp $\widetilde k$ to $[1/\kappa,\kappa n]$, which preserves every correct estimate, and set
$$
p\coloneqq2^{-\lceil\log(4\kappa\widetilde k)\rceil}, \qquad \frac1{8\kappa\widetilde k}<p\le\frac1{4\kappa\widetilde k}.
$$
Thus $p$ is dyadic, and whenever $k/\kappa\le\widetilde k\le\kappa k$,
\begin{equation}
\frac1{8\kappa^2}<pk\le\frac14, \qquad p\le\frac14.
\label{eq:qupper-23}
\end{equation}
The clamping also bounds the number of possible choices of $p$ and the cost of every branch, including those on which preprocessing fails.

\paragraph{An Empty-Or-Singleton Experiment.}
The implemented experiment may also return the overflow marker after a recovery failure; this is included in its error probability.
Retain each coordinate independently with probability $p$, and let $D$ be the set of retained coordinates for which $\alpha(x_i)\ne\beta(y_i)$. Consider the ideal experiment with output $0$ if $D=\varnothing$, output $e=(x_i,y_i)$ if $D=\{i\}$, and output $\mathrm{overflow}$ if $|D|\ge2$. Let $P_0$ and $P_e$ be the probabilities of outputs $0$ and $e\in E$. Independence gives
$$
P_0=(1-p)^k, \qquad P_e=T_ep(1-p)^{k-1}.
$$
Under these bounds,
\begin{equation}
T_e=\frac{1-p}{p}\frac{P_e}{P_0}, \qquad P_0\ge1-pk\ge\frac34.
\label{eq:qupper-24}
\end{equation}
This identity recovers the unnormalized counts while keeping the denominator uniformly bounded away from zero.

We implement the experiment with error at most $\eta$, where $0<\eta<1/4$ will be fixed below. Let $A$ be a uniform random matrix in $\mathbb F_2^{b\times n}$, chosen independently of the retained set, with
$$
b\coloneqq\left\lceil\log\frac{n+1}{\eta}\right\rceil.
$$
Both parties mask their locally recoded binary strings by the retained set. Alice sends the syndrome of her masked string, and Bob adds his own syndrome to obtain $A\mathbf1_D$. Bob searches the candidate set $\{0,e_1,\ldots,e_n\}$ for a unique vector having this syndrome, where $e_i$ is the $i$th standard basis vector of $\mathbb F_2^n$. If there is no unique match, he outputs overflow. A unique match at zero produces output $0$. A match at $e_i$ is followed by communication of $i$ and the original symbol $x_i$, so that the ordered pair $(x_i,y_i)$ is known to Bob. He reports this pair if it belongs to $E$, and reports overflow otherwise.

Conditional on any retained set, Lemma~\ref{lem:linear-hash} shows that any fixed wrong candidate collides with $\mathbf1_D$ with probability $2^{-b}$. A union bound over the $n+1$ candidates bounds the error by $\eta$. This argument also handles overflow: if $|D|\ge2$, then $\mathbf1_D$ is outside the candidate set, and a false empty or singleton report still has probability at most $\eta$. Pad the messages with a dummy coordinate and a dummy symbol whenever needed. The syndrome, coordinate, and symbol messages then have fixed lengths, with total communication
\begin{equation}
O_q\bigl(\log n+\log(1/\eta)\bigr).
\label{eq:qupper-25}
\end{equation}

\paragraph{Compressing and Coherently Using the Sampling Randomness.}
Each random tape specifies the retained set and the matrix $A$. For every possible $E,p$, apply Lemma~\ref{lem:finite-tape} to the syndrome experiment, whose output set has at most $q^2+2$ elements, and fix a table of
\begin{equation}
K=O\!\left(\eta^{-2}n\log q\right)
\label{eq:qupper-26}
\end{equation}
full tapes, with $K$ a power of two. Each table approximates every output probability to additive error $\eta$ simultaneously for all input pairs; its guarantee therefore survives the input-dependent choice of $p$ in preprocessing. These tables belong to the protocol description, with no efficient construction or local decoding assumed. A communicated table index has length $\log K=O_q(\log n+\log(1/\eta))$ and supplies the randomness for the coherent experiment below.

Fix the classical preprocessing transcript for the remainder of this construction; its records and the input registers are kept unchanged. For a fixed input pair, let $Z(s)$ be the deterministic experiment output on table entry $s$. We identify each output label with its fixed-length binary encoding. We use Bennett's reversible computation and uncomputation method~\cite{Bennett73}, with the communication accounted for explicitly. Alice prepares a uniform superposition of table indices, coherently copies the index in the computational basis into an auxiliary register, and sends this register to Bob. The parties reversibly execute the fixed-message protocol using the random tape specified by the index, retaining the workspace needed to reverse each step. Bob sends the output register to Alice, who XORs its contents into a separate target register. They then reverse the output transmission and the protocol computation, leaving the target register unchanged. Bob returns the index-copy register to Alice, who reverses the copying operation to restore it to zero. All auxiliary registers return to their initial states. This implements the following map on Alice's registers:
\begin{equation}
|s\rangle|w\rangle \longmapsto |s\rangle|w\oplus Z(s)\rangle.
\label{eq:qupper-27}
\end{equation}
Every auxiliary register returns to zero. This is an exact unitary implementation of the table-defined function for the fixed input pair; the sampling and recovery errors concern its output distribution relative to the ideal experiment. Its inverse and its controlled version have the same asymptotic communication cost as \eqref{eq:qupper-25}, including the cost of communicating the table index. To implement a controlled version, Alice sends a computational-basis copy of the control qubit along with the index copy. Both parties use a fixed communication schedule and condition their local computations on the control qubit or its copy. Bob returns both copies to Alice, who reverses the copying operations. Communicating the additional control qubit adds only $O(1)$ qubits per call. Composing this map with Alice's local preparation of a uniform superposition of table indices gives the state-preparation unitary required by amplitude estimation. The initial-state reflection and the output-predicate reflection can both be performed on Alice's registers, since Bob's auxiliary registers are restored to zero between calls. In particular, the construction uses no free distributed reflection and no prior entanglement.

For each event with $\widetilde k>0$, define the actual outcome probabilities for the fixed input pair and preprocessing transcript by
$$
\bar P_z\coloneqq \frac1K\bigl|\{s\in[K]:Z(s)=z\}\bigr|, \qquad z\in\{0\}\cup E.
$$
These are the output probabilities of the coherent experiment. The syndrome recovery error and the finite-table approximation each contribute at most $\eta$, so
$$
|\bar P_z-P_z|\le2\eta.
$$

\paragraph{Estimating Outcome Probabilities.}
Apply Lemma~\ref{lem:sqrt-amplitude} to estimate $\bar P_0$ and every $\bar P_e$ for each event with $\widetilde k>0$. There are at most $|\mathcal E|(q^2+1)$ estimates. For each outcome, let $\widehat P_z$ be the median of an odd number of independent repetitions depending only on $q$. Each repetition succeeds with probability at least $8/\pi^2>1/2$. Hoeffding's inequality and a union bound therefore give, with probability at least $19/20$,
$$
\bigl|\sqrt{\widehat P_z}-\sqrt{\bar P_z}\bigr|
\le \frac{C_q}{t}
$$
simultaneously for all relevant outcomes and events. This guarantee holds conditional on any fixed preprocessing transcript and for every fixed input pair. Since the square root is monotone, medians preserve the square-root error bound.

Set $\tau\coloneqq C_q/t+\sqrt{2\eta}$. Using $|\sqrt u-\sqrt v|\le\sqrt{|u-v|}$ for nonnegative $u,v$, the triangle inequality gives
\begin{equation}
\begin{aligned}
\bigl|\sqrt{\widehat P_z}-\sqrt{P_z}\bigr|
&\le \bigl|\sqrt{\widehat P_z}-\sqrt{\bar P_z}\bigr| +\bigl|\sqrt{\bar P_z}-\sqrt{P_z}\bigr|\\
&\le \frac{C_q}{t}+\sqrt{2\eta} =\tau
\end{aligned}
\label{eq:qupper-29}
\end{equation}
simultaneously for all relevant outcomes and events. For the following count-error bounds, assume also that every estimate $\widetilde k$ satisfies the stated accuracy bounds. Clip $\widehat P_0$ to $[1/2,1]$; since $P_0\ge3/4$ under these bounds, this cannot increase $|\sqrt{\widehat P_0}-\sqrt{P_0}|$. Define the nonnegative, not necessarily integer, count estimates by
$$
\widehat W_e^E\coloneqq \frac{1-p}{p}\frac{\widehat P_e}{\widehat P_0}.
$$
Using \eqref{eq:qupper-24}, split each square-root difference as
$$
\sqrt{\widehat W_e^E}-\sqrt{T_e} =\sqrt{\frac{1-p}{p}}\left[ \frac{\sqrt{\widehat P_e}-\sqrt{P_e}}{\sqrt{\widehat P_0}} +\sqrt{P_e}\left( \frac1{\sqrt{\widehat P_0}}-\frac1{\sqrt{P_0}}\right) \right].
$$
Since $P_0\ge3/4$, $\widehat P_0\ge1/2$, and $|\sqrt{\widehat P_0}-\sqrt{P_0}|\le\tau$, the difference of reciprocal square roots is at most $2\tau$. Apply $(u+v)^2\le2u^2+2v^2$, use $|E|\le q^2$ and $\sum_eP_e\le1$, and sum to obtain
\begin{equation}
\sum_{e\in E}(\sqrt{\widehat W_e^E}-\sqrt{T_e})^2 \le\frac{(4q^2+8)\tau^2}{p}.
\label{eq:qupper-30}
\end{equation}
The denominator in \eqref{eq:hellinger} is at least $k$. Consequently, \eqref{eq:qupper-23} implies
\begin{equation}
H(J_{\widehat W^E},J_T^E)^2 \le\frac{(4q^2+8)\tau^2}{pk} \le8\kappa^2(4q^2+8)\tau^2.
\label{eq:qupper-31}
\end{equation}
Choose $t\coloneqq\lceil A_q/h\rceil$ and $\eta\coloneqq c_qh^2$, with respectively sufficiently large and sufficiently small positive constants depending only on $q$, where $A_q$ is chosen after the constant $C_q$ in \eqref{eq:qupper-29}. Equations \eqref{eq:qupper-29}--\eqref{eq:qupper-31} then make every event distance at most $h/8$. Events with $k=0$ always have exact zero count vectors.

\paragraph{Decoding and Communication Cost.}
Alice applies Lemma~\ref{lem:type-decoder} to the estimates $\widehat J_E\coloneqq J_{\widehat W^E}$. She enumerates all promised types with the known margins, approximates each Hellinger distance to additive error $h/32$, and minimizes the resulting score. The lemma applies with $\varepsilon=h/8$ and $\zeta=h/32$, so she sends the label of a minimizer to Bob. As in the classical protocol, this is finite local computation, and the guarantee holds for the entire promise without a union bound over types.

The probability that some estimate $\widetilde k$ violates its accuracy bounds or some amplitude estimate fails is at most $1/20+1/20<1/3$. The number of events and estimated outcomes depends only on $q$, and every coherent call and repetition has a predetermined cost. By \eqref{eq:qupper-22}, $\log(1/\eta)=O_q(\log n)$. Thus each coherent experiment costs $O_q(\log n)$ qubits, and there are $O_q(h^{-1})$ calls in total. Preprocessing costs $O_q(\log n)$ bits and is absorbed in the total bound. This proves the $O_q(h^{-1}\log n)$ bound in \eqref{eq:qupper-20}; sending a full input gives the $O_q(n)$ alternative.
\end{proof}

\clearpage
\section{Uniform Polynomial Approximation of Quantum Protocols}
\label{sec:polynomial}

This appendix proves Lemma~\ref{lem:profile}, the uniform polynomial approximation used in the single-cycle lower bound. We fix the margins and require large counts on a connected spanning subgraph of the active symbols.

The reusable steps are to take a fixed power of the averaging operator so that every one-coordinate marginal atom is bounded below, and to truncate the acceptance matrix using degree projections determined by the fixed margins; the resulting polynomial is uniform over all eligible types.

\subsection{Multislices and Conditional Averaging Operators}

For a finite alphabet $\Sigma$ and an integer histogram $\rho=(\rho_a)_{a\in\Sigma}$ of total $N\ge1$, the \emph{multislice} is the set of strings with that histogram:
$$
X_\rho\coloneqq\{x\in\Sigma^N:|\{i:x_i=a\}|=\rho_a
\text{ for every }a\in\Sigma\}.
$$
Equip it with the uniform probability measure and let $X\coloneqq X_\rho$. A symbol is active if its count is positive; inactive symbols can be discarded. A function depending on a specified set of at most $d$ coordinates is a $d$-junta. Let $V_{X,\le d}\subseteq L^2(X)$ be the span of these functions, $P_{X,d}$ its orthogonal projection, and $V_{X,>d}\coloneqq V_{X,\le d}^{\perp}$. These are the standard degree spaces on a multislice used in~\cite{BKLM}.

We record the content of the multislice contraction theorem that will be used. A multislice is $\alpha$-balanced if each active symbol has frequency at least $\alpha$. An invariant probability measure on a product of multislices is $\alpha$-admissible if every nonzero atom of its one-coordinate marginal has probability at least $\alpha$ \cite{BKLM}. Invariance here is under common permutations of coordinates. For a measure $\nu$ on $X\times X$, connectivity means that the bipartite graph on two copies of $X$, with an edge for every pair of positive probability, is connected \cite{BKLM}. This connectivity condition concerns complete length-$N$ strings, rather than individual alphabet symbols.

The product analogue of $\nu$ is the law in which the coordinate pairs are independent and each has the one-coordinate marginal of $\nu$. For two invariant laws on strings of length $N$ over the same finite alphabet, an $(\alpha,\zeta)$-coupling in the sense of \cite{BKLM} is a joint law of $(W,W')$ with the specified marginals, invariant under common coordinate permutations, such that for every coordinate $i$ and every $\xi>0$,
$$
\Prb[W_i\ne W'_i]\le\zeta,\qquad
\Prb[|\{i:W_i\ne W'_i\}|\ge\xi N]
\le\alpha^{-1}e^{-\alpha\xi^2N}.
$$
To couple laws of several strings, regard the entire coordinate tuple as one symbol in the product alphabet. We construct the coupling in this sense below; its component couplings satisfy these bounds as well.

\begin{theorem}[Multislice Contraction, {\cite{BKLM}}]
\label{thm:bklm}
For fixed alphabet size and balance/admissibility constant $\alpha>0$, there are constants $K<\infty$ and $\eta>0$ as follows. Let $X$ be an $\alpha$-balanced multislice and let $\nu$ be an \textit{$\alpha$-admissible}, connected measure on $X\times X$ with uniform marginals and an $(\alpha,\zeta)$-coupling to its product analogue for some $\zeta>0$. Its conditional averaging operator, defined by
$$
(T_\nu f)(x')\coloneqq\E_{(x,x')\sim\nu}[f(x)\mid x'],
$$
satisfies
$$
\|T_\nu f\|_2\le K(1+\eta)^{-d}\|f\|_2
\qquad(f\in V_{X,>d},\ d\ge1).
$$
\end{theorem}

Theorem~\ref{thm:bklm} is the cited external result. Below we verify its hypotheses for the operators used in our argument, using couplings with $\zeta=O_{q,\alpha}(N^{-1/2})$. The cited theorem imposes no additional smallness condition on $\zeta$.

Fix row and column margins $\rho,\sigma$, and let $X\coloneqq X_\rho$, $Y\coloneqq X_\sigma$. For a type $S$ with these margins, let $(x,y)$ be uniform on all input pairs of type $S$. Both marginals are uniform: the permutation group is transitive on each multislice and leaves the joint law invariant. Define
\begin{equation}\label{eq:type-operator}
B_S:L^2(X)\longrightarrow L^2(Y),\qquad (B_S f)(y)\coloneqq\E[f(x)\mid y].
\end{equation}
Here $\|\cdot\|_{2\to2}$ denotes the operator norm induced by the $L^2$ norms under the uniform measures. The adjoint $B_S^*$ averages in the reverse direction. If $f$ depends only on positions in $I$, then $B_S f$ also depends only on positions in $I$: for each column symbol, the row symbols at its positions are sampled without replacement with multiplicities prescribed by $S$. The law on $I$ uses only the column symbols on $I$. The same argument applies to $B_S^*$. Consequently both the low-degree spaces and their orthogonal complements are preserved, and
\begin{equation}\label{eq:projections-commute}
P_{Y,d}B_S=B_SP_{X,d}.
\end{equation}
For example, preservation of the orthogonal complement follows from $\langle B_Sf,g\rangle=\langle f,B_S^*g\rangle=0$ when $f\in V_{X,>d}$ and $g\in V_{Y,\le d}$.

\subsection{Operator Contraction from Large Joint-Type Counts}

We use the following standard martingale concentration bound.

\begin{lemma}[Azuma--Hoeffding Inequality, {\cite{Azuma67,Hoeffding63}}]\label{lem:azuma}
Let $(M_0,\ldots,M_t)$ be a real martingale and suppose $|M_j-M_{j-1}|\le b_j$ almost surely, where the $b_j$ are deterministic. Then, for $u>0$,
$$
\Prb[|M_t-M_0|\ge u] \le2\exp\left(-\frac{u^2}{2\sum_{j=1}^t b_j^2}\right).
$$
If every $b_j=0$, the difference is zero and the tail bound is interpreted accordingly.
\end{lemma}

We now prove the contraction estimate for type-averaging operators.

\begin{lemma}\label{lem:operator-contraction}
Fix an integer $q\ge2$ and $\beta>0$. Suppose a type $S$ of total $N\ge1$, with at most $q$ active symbols on each side, has a connected spanning bipartite subgraph $G$ on its active row and column symbols such that
\begin{equation}\label{eq:large-entries}
S_e\ge\beta N\qquad\text{for every }e\in G.
\end{equation}
Here $S_e\coloneqq S_{ab}$ for an edge $e=(a,b)$, and $e\in G$ means that $e$ is an edge of $G$. Then there are constants $K^{\mathrm{op}}_{q,\beta}\ge1$ and $\gamma_{q,\beta}>0$ such that, for every integer $d\ge0$,
\begin{equation}\label{eq:operator-contraction}
\|B_S(I-P_{X,d})\|_{2\to2} \le K^{\mathrm{op}}_{q,\beta}e^{-\gamma_{q,\beta}d}.
\end{equation}
The constants are uniform over $N$, the margins, and eligible $S,G$.
\end{lemma}
\begin{proof}
Discard inactive row and column symbols, so every margin appearing in a denominator below is positive. If there is only one active row symbol, $X$ is a singleton and the left-hand side is zero. Otherwise let $L_S\coloneqq B_S^*B_S$. This is a positive semidefinite, self-adjoint Markov operator on $X$, describing the transition $x\to y\to x'$. Every active margin is at least $\beta N$, since every vertex is incident to an edge of $G$.

At a fixed coordinate, the row-symbol transition matrix for $L_S$ is
\begin{equation}\label{eq:one-coordinate-transition}
H(a,a')\coloneqq\sum_b\frac{S_{ab}}{\rho_a}\frac{S_{a'b}}{\sigma_b}.
\end{equation}
Indeed, the first factor selects the intermediate column symbol conditional on the initial row symbol, and the second selects the final row symbol. This formula holds conditional on the current length-$N$ string, so it can be iterated for repeated steps.

Connect two row symbols if they share a neighbor in $G$. This row-symbol graph is connected. Along any such adjacency the corresponding summand in~\eqref{eq:one-coordinate-transition} is at least $\beta^2$, because both numerators are at least $\beta N$ and both denominators are at most $N$. Every row symbol also has a self-transition of probability at least $\beta^2$. Let $r\coloneqq q$. A path between two row symbols has length at most $q-1$; pad it with self-transitions to length $r$. It follows that
\begin{equation}\label{eq:positive-marginal}
H^r(a,a')\ge\beta^{2q},\qquad
\theta_{aa'}\coloneqq\frac{\rho_a}{N}H^r(a,a')\ge\beta^{2q+1}.
\end{equation}
Let $\nu_r$ be the joint distribution of $(x,x')$ obtained by choosing $x$ uniformly from $X$ and running $r$ steps of the Markov chain with transition operator $L_S$ to obtain $x'$. It has uniform multislice marginals, is invariant under common coordinate permutations, and is symmetric under exchanging the endpoints. Its one-coordinate marginal is $\theta$. Thus, after taking the fixed power, every row-pair atom has a lower bound independent of $N$. This step is needed because entries outside $G$ in $S$ may be small.

We next verify that the bipartite support graph on two copies of $X$ is connected. Suppose distinct row symbols $a,a'$ share a column neighbor $b$ in $G$, and positions $i,j$ in a string $x$ carry $a,a'$. There is a compatible column string $y$ with $y_i=y_j=b$: assign one occurrence from each of the positive integer cells $S_{ab}$ and $S_{a'b}$ to these positions, and fill the remaining positions according to $S$. Swapping $x_i,x_j$ produces a string $x'$ also compatible with that same $y$. Therefore the transition $L_S(x,x')$ has positive probability.

These swaps connect the entire multislice. To exchange two specified occurrences of different symbols, take a simple path between their symbols in the row-symbol graph and choose one occurrence of each intermediate symbol. Track these chosen occurrences as they move. Swap the first and second chosen occurrences, then the second and third, and continue along the path. Then undo these swaps in reverse order, omitting the final swap of the forward sequence. This exchanges the endpoint occurrences and restores all intermediate occurrences to their original positions. Each swap involves symbols adjacent in the row-symbol graph, regardless of their current positions, so it is an allowed transition. Thus any two occurrences of different symbols can be exchanged, and such exchanges connect all strings with the fixed histogram.

Moreover, $L_S(x,x)>0$ for every $x$: after moving from $x$ to any compatible $y$, there is positive probability of returning to $x$. Any positive-probability transition of $L_S$ can therefore be extended to $r$ steps by staying at its endpoint for the remaining $r-1$ steps. Hence the support of $L_S^r$ contains that of $L_S$, including all self-loops. The bipartite support graph on two copies of $X$ is consequently connected: each self-loop joins the two copies of the same string, and the other supported transitions connect the different strings. Thus $\nu_r$ has the required connected bipartite support graph.

It remains to construct the coupling. Fix one pair $(x_0,y_0)$ of type $S$. Let $K_X,K_Y\le S_N$ be the stabilizers of $x_0,y_0$, respectively; they permute positions separately within the row-symbol and column-symbol classes. Let $g$ be uniform in $S_N$, and independently choose uniform $u_j\in K_X$, $v_j\in K_Y$ for $1\le j\le r$.
Then
\begin{equation}\label{eq:group-endpoints}
\bigl(gx_0,\;gu_1v_1u_2v_2\cdots u_rv_rx_0\bigr)
\end{equation}
has law $\nu_r$. To check this, set $h_0\coloneqq g$ and $h_j\coloneqq h_{j-1}u_jv_j$. Conditional on the current representative $h_{j-1}$, the intermediate string $h_{j-1}u_jy_0$ is uniform among column strings compatible with $h_{j-1}x_0$: applying a uniform element of $K_X$ to $y_0$ gives a uniform compatible column string, since all fibers of this group action have the same size. Conditional also on $u_j$, applying a uniform element of $K_Y$ to $x_0$ then makes $h_{j-1}u_jv_jx_0$ uniform among row strings compatible with this intermediate column string. Iterate these two steps $r$ times.

Let $R_{aa'}$ count the endpoint symbol pairs in~\eqref{eq:group-endpoints}. The common permutation $g$ does not affect $R$. Expose the images of the independent block permutations $u_j,v_j$, one at a time, using at most $2rN$ exposures. Two possible images at a particular exposure have completions in bijection by a transposition of the remaining images. In the final product of permutations, this transposition becomes a conjugate transposition. It changes the final string at at most two positions, and thus changes any fixed $R_{aa'}$ by at most two. Coupling the remaining completions by this bijection shows that the associated Doob martingale differences have absolute value at most two. There are at most $2rN$ martingale increments, each bounded in absolute value by two, so
$$
2\sum_j b_j^2\le2(2rN)\cdot2^2=16rN.
$$
Lemma~\ref{lem:azuma} therefore gives
\begin{equation}\label{eq:endpoint-concentration}
\Prb[|R_{aa'}-N\theta_{aa'}|\ge t] \le2\exp\bigl(-t^2/(16rN)\bigr).
\end{equation}
The mean is $N\theta_{aa'}$ by~\eqref{eq:positive-marginal} and coordinate symmetry.

For the product analogue, let $R'$ be its table of coordinate-pair counts. These counts have a multinomial distribution with cell probabilities $\theta$. Exposing its $N$ independent coordinate pairs gives martingale differences bounded by one and hence
\begin{equation}\label{eq:product-concentration}
\Prb[|R'_{aa'}-N\theta_{aa'}|\ge t] \le2\exp\bigl(-t^2/(2N)\bigr).
\end{equation}
Draw $R,R'$ independently. Match as many occurrences of each ordered symbol pair as possible, match the remaining occurrences arbitrarily, and apply a common uniform permutation to the two lists of coordinate pairs. Conditional on either table, each marginal string of ordered symbol pairs is uniform among its arrangements. This holds for $\nu_r$ by permutation invariance and for the product law by its product formula. The construction therefore has precisely the two desired marginal distributions.

The number $D$ of mismatched coordinate pairs in this coupling is
\begin{equation}\label{eq:coupling-distance}
D=\frac12\sum_{a,a'}|R_{aa'}-R'_{aa'}|.
\end{equation}
There are at most $q^2$ entries. Integrating the tails in \eqref{eq:endpoint-concentration}--\eqref{eq:product-concentration} and using the triangle inequality in~\eqref{eq:coupling-distance} yields $\E D=O_q(\sqrt N)$. For example, $\int_0^\infty2e^{-t^2/(16rN)}\,dt=4\sqrt{\pi rN}$. If $D\ge\xi N$, one of the at most $2q^2$ deviations $|R_{aa'}-N\theta_{aa'}|$, $|R'_{aa'}-N\theta_{aa'}|$ is at least $\xi N/q^2$. A union bound consequently gives
\begin{equation}\label{eq:coupling-tail}
\Prb[D\ge\xi N]\le4q^2 \exp\bigl(-\xi^2N/(16q^5)\bigr).
\end{equation}
The common random permutation makes the coupling invariant. Each coordinate has mismatch probability $\E D/N=O_q(N^{-1/2})$. A mismatch in either component string is bounded by the mismatch in the coordinate pair, so the component couplings satisfy the same bounds.

Choose a positive constant
$$
\alpha\le\min\{\beta,\beta^{2q+1},(4q^2)^{-1},(16q^5)^{-1},1/2\}.
$$
The multislice is $\alpha$-balanced, $\nu_r$ is $\alpha$-admissible and connected, and the constructed coupling has the required $(\alpha,\zeta)$ properties, with $\zeta\le\min\{1,C_qN^{-1/2}\}$. Since $\nu_r$ is symmetric, its conditional averaging operator satisfies $T_{\nu_r}=L_S^r$. Thus Theorem~\ref{thm:bklm} applies and gives
$$
\bigl\|L_S^r\big|_{V_{X,>d}}\bigr\|_{2\to2}
\le K(1+\eta)^{-d}.
$$
All constants depend only on $q,\beta$; one can take the worst of the finitely many possible active alphabet sizes. The case $d=0$, if excluded from the convention for natural numbers in the cited statement, is covered by increasing $K$ to at least one, since a Markov averaging operator is an $L^2$ contraction by conditional Jensen's inequality.

By~\eqref{eq:projections-commute}, $V_{X,>d}$ is invariant under $L_S$. Positivity and self-adjointness, followed by the spectral theorem, give
$$
\bigl\|B_S\big|_{V_{X,>d}}\bigr\|_{2\to2}^{2r}
=\bigl\|L_S^r\big|_{V_{X,>d}}\bigr\|_{2\to2}.
$$
Taking the $2r$-th root proves~\eqref{eq:operator-contraction}
with
$$
K^{\mathrm{op}}_{q,\beta}\coloneqq K^{1/(2r)}, \qquad \gamma_{q,\beta}\coloneqq\frac{\ln(1+\eta)}{2r}.
$$
\end{proof}

\subsection{From Acceptance Matrices to Joint-Type Polynomials}

We use the trace-norm consequence of the Linial--Shraibman $\gamma_2$ factorization bound~\cite{LinialShraibman09}. For a matrix $A$, the factorization norm $\gamma_2(A)$ is the minimum, over factorizations $A=UV$, of the product of the largest Euclidean row norm of $U$ and the largest Euclidean column norm of $V$. The following Frobenius-norm formulation is stated in~\cite{SS}.

\begin{theorem}[Quantum Acceptance-Matrix Factorization ($\gamma_2$ Method),
{\cite{SS}}]\label{thm:factorization} Let $X,Y$ be finite input sets. If $A(y,x)$ is the acceptance matrix of a protocol communicating $c$ qubits, with arbitrary prior entanglement allowed, then $A=UV$ for real matrices satisfying
$$
\|U\|_{\mathrm F}\le2^c\sqrt{|Y|},\qquad
\|V\|_{\mathrm F}\le2^c\sqrt{|X|}.
$$
Here $\|\cdot\|_{\mathrm F}$ denotes the Frobenius norm.
\end{theorem}

For our normalization, this gives the following trace-norm bound, the form of the method used in the polynomial approximation below.

\begin{corollary}[Normalized Trace-Norm Bound]\label{cor:trace-bound} Under the hypotheses of Theorem~\ref{thm:factorization}, let $\overline A\coloneqq A/\sqrt{|X||Y|}$. Then
\begin{equation}\label{eq:trace-bound}
\|\overline A\|_*\le4^c,
\end{equation}
where $\|\cdot\|_*$ denotes the trace norm (also called the nuclear norm), the sum of singular values. The same bound holds for convex mixtures of acceptance matrices of protocols with cost at most $c$.
\end{corollary}
\begin{proof}
Let $UV$ be a sum of outer products of columns of $U$ and rows of $V$. The trace norm of an outer product is the product of its Euclidean norms. The triangle inequality and Cauchy--Schwarz therefore give $\|UV\|_*\le\|U\|_{\mathrm F}\|V\|_{\mathrm F}$. By Theorem~\ref{thm:factorization},
$$
\|\overline A\|_*
\le\frac{\|U\|_{\mathrm F}\|V\|_{\mathrm F}}{\sqrt{|X||Y|}}
\le4^c.
$$
Convexity of the trace norm gives the assertion for mixtures, including public randomization.
\end{proof}

\begin{lemma}[Probabilities of Local Patterns]\label{lem:local-pattern}
Let $(x,y)$ be uniform among input pairs of joint type $S$ and total length $N$. Fix $\ell$ distinct coordinate positions and specify an ordered symbol pair at each position. If $m_{ab}$ of these positions specify $(a,b)$, then the probability of this pattern is
\begin{equation}\label{eq:falling-factorial}
\frac{\prod_{a,b}\fall{S_{ab}}{m_{ab}}}{\fall{N}{\ell}},
\qquad
\ell=\sum_{a,b}m_{ab},\qquad 0\le\ell\le N,
\end{equation}
where
$$
\fall{z}{k}\coloneqq z(z-1)\cdots(z-k+1)\quad(k\ge1),
\qquad
\fall{z}{0}\coloneqq1.
$$
For fixed $N$ and a fixed pattern, this probability is a polynomial of total degree at most $\ell$ in the entries of $S$.
\end{lemma}
\begin{proof}
There are $N!/\prod_{a,b}S_{ab}!$ input pairs of type $S$. If $S_{ab}\ge m_{ab}$ for every pair, fixing the specified positions leaves $(N-\ell)!/\prod_{a,b}(S_{ab}-m_{ab})!$ possible completions. Their ratio is \eqref{eq:falling-factorial}. If some $S_{ab}<m_{ab}$, the pattern is impossible and its probability is zero; the same formula still holds because the corresponding falling factorial is zero. The denominator is a fixed positive number because $0\le\ell\le N$, and the numerator has total degree $\sum_{a,b}m_{ab}=\ell$.
\end{proof}

We now combine the operator contraction and trace-norm bounds with the local-pattern formula to prove the uniform approximation result used in Lemma~\ref{lem:single-cycle}.

\begin{profilelemma}[Uniform Polynomial Approximation, Restated]
Fix an integer $q\ge2$, $\beta>0$, total length $N\ge1$, and row and column margins with at most $q$ active symbols per side. Let $G$ be a connected spanning bipartite subgraph on the active symbols. For any $c$-qubit protocol with prior entanglement, let $p(S)$ denote its acceptance probability averaged uniformly over all input pairs of joint type $S$. For each $\varepsilon\in(0,1)$, there is a single real polynomial $P$ in the entries of $S$, of total degree at most
$$
K_{q,\beta}\bigl(c+\log(1/\varepsilon)\bigr),
$$
such that $|p(S)-P(S)|\le\varepsilon$ for every integer type with the fixed margins satisfying
$$
S_e\ge\beta N\qquad\text{for every }e\in G.
$$
Entries outside $G$ may be zero. Here $K_{q,\beta}>0$ depends only on $q$ and $\beta$.
\end{profilelemma}
\begin{proof}
If $\varepsilon\ge1/2$, take $P\coloneqq1/2$; if no integer type satisfies the hypotheses, take $P\coloneqq0$. Otherwise assume $0<\varepsilon<1/2$ and let $X\coloneqq X_\rho$, $Y\coloneqq X_\sigma$ for the fixed-margin multislices, equipped with uniform probability measures. Thus
$$
\langle u,v\rangle_Y=\frac1{|Y|}\sum_{y\in Y}u(y)v(y),
\qquad \|u\|_{L^2(Y)}^2=\langle u,u\rangle_Y,
$$
and analogously on $X$. Lemma~\ref{lem:operator-contraction} supplies $K_0\ge1$ and $\gamma>0$, depending only on $q,\beta$, such that
$$
\|B_S(I-P_{X,d})\|_{2\to2}\le K_0e^{-\gamma d}
$$
for every eligible $S$ and every integer $d\ge0$.

Let $A(y,x)$ be the protocol's acceptance matrix and let $\overline A\coloneqq A/\sqrt{|X||Y|}$. Take a real Euclidean singular-value decomposition $\overline A=\sum_j s_j\mathbf u_j\mathbf v_j^{\mathsf T}$, with unit singular vectors, and set $u_j(y)\coloneqq\sqrt{|Y|}\,\mathbf u_j(y)$ and $v_j(x)\coloneqq\sqrt{|X|}\,\mathbf v_j(x)$. Then
$$
A(y,x)=\sum_j s_ju_j(y)v_j(x),\qquad \|u_j\|_{L^2(Y)}=\|v_j\|_{L^2(X)}=1,\qquad \sum_j s_j=\|\overline A\|_*\le4^c,
$$
where the last inequality is Corollary~\ref{cor:trace-bound}. In particular, $s_j\ge0$ are the singular values of the normalized matrix $\overline A$. Let $\E_S$ denote expectation under the uniform type-$S$ distribution. Its marginals are uniform, so conditional expectation gives
$$
\E_S[u(y)v(x)]=\langle u,B_Sv\rangle_Y, \qquad p(S)=\sum_j s_j\langle u_j,B_Sv_j\rangle_Y.
$$

For an integer $0\le d\le N$, define
\begin{equation}\label{eq:low-degree-acceptance}
P_d(S)\coloneqq\sum_j s_j\, \E_S\!\left[(P_{Y,d}u_j)(y)(P_{X,d}v_j)(x)\right].
\end{equation}
The commuting-projection identity~\eqref{eq:projections-commute} and $P_{X,d}^2=P_{X,d}$ imply $P_{Y,d}B_SP_{X,d}=B_SP_{X,d}$. Since $P_{Y,d}$ is self-adjoint,
$$
\langle P_{Y,d}u_j,B_SP_{X,d}v_j\rangle_Y =\langle u_j,P_{Y,d}B_SP_{X,d}v_j\rangle_Y =\langle u_j,B_SP_{X,d}v_j\rangle_Y.
$$
Consequently the truncation error has the exact expression
$$
p(S)-P_d(S)
=\sum_j s_j\langle u_j,B_S(I-P_{X,d})v_j\rangle_Y.
$$
By Cauchy--Schwarz and the operator norm bound,
\begin{align}
|p(S)-P_d(S)| &\le \sum_j s_j \|u_j\|_{L^2(Y)} \|B_S(I-P_{X,d})\|_{2\to2} \|v_j\|_{L^2(X)} \notag\\
&\le K_0e^{-\gamma d}\sum_j s_j \le 4^cK_0e^{-\gamma d}.
\label{eq:truncation-error}
\end{align}
Here we use $\|u_j\|_{L^2(Y)}=\|v_j\|_{L^2(X)}=1$, the contraction estimate, and $\sum_j s_j\le4^c$. The bound holds for every type $S$ satisfying the hypotheses.

We next represent~\eqref{eq:low-degree-acceptance} by one polynomial in $S$. Each projected function is a finite linear combination of $d$-juntas, so each product in its expansion depends on at most $\min\{2d,N\}$ coordinate positions. Expand these products into indicators of specified symbol patterns. A condition on Alice's symbol alone can be written as
$$
\mathbf1_{\{x_i=a\}} =\sum_b\mathbf1_{\{(x_i,y_i)=(a,b)\}}.
$$
The same applies to a condition on Bob's symbol alone. By Lemma~\ref{lem:local-pattern}, the expectation of the indicator of each resulting ordered-pair pattern on $\ell$ positions is a polynomial of total degree at most $\ell\le\min\{2d,N\}$. The fixed protocol and margins determine $A$ and allow us to choose its singular-value decomposition independently of $S$. For each fixed $d$, the projections and expansion coefficients are also independent of $S$. Taking the corresponding finite linear combination of these pattern-probability polynomials gives a single real polynomial $P_d$ of degree at most $\min\{2d,N\}$. This polynomial represents~\eqref{eq:low-degree-acceptance} on every integer type with these margins. The error bound~\eqref{eq:truncation-error} holds on the eligible types.

Finally, set
$$
d_0\coloneqq\left\lceil \gamma^{-1}\bigl(c\ln4+\ln(K_0/\varepsilon)\bigr) \right\rceil,\qquad d\coloneqq\min\{d_0,N\},\qquad P\coloneqq P_d.
$$
If $d_0\le N$, then~\eqref{eq:truncation-error} is at most $4^cK_0e^{-\gamma d_0}\le\varepsilon$. If $d_0>N$, every function on either multislice is an $N$-junta, so $P_{X,N}=P_{Y,N}=I$ and \eqref{eq:low-degree-acceptance} gives $P_N(S)=p(S)$ on every type with the fixed margins. In either case,
$$
\deg P\le2d_0 \le2\gamma^{-1}\bigl(c\ln4+\ln K_0+\ln(1/\varepsilon)\bigr)+2 \le K_{q,\beta}\bigl(c+\log(1/\varepsilon)\bigr),
$$
after increasing $K_{q,\beta}$, since $K_0,\gamma$ depend only on $q,\beta$ and $\log(1/\varepsilon)\ge1$.
\end{proof}

\clearpage

\section{Auxiliary Estimates and Alphabet Constants}
\label{sec:auxiliary}

This appendix collects the auxiliary proofs for the upper and lower bounds: scale estimation, batch recovery, sampled-count error, and the adjacent-integer degree estimate. It also tracks the constants in Proposition~\ref{prop:alphabet-upper} and compares our parameter with the binary jump measure of Guan et al.~\cite{GHYY}. The classical protocol and its overall error analysis remain in Section~\ref{sec:classical}; quantum resource accounting appears in Appendix~\ref{sec:quantum-upper}.

\subsection{Disagreement-Count Scale Estimation}\label{sec:aux-scale}

We restate the scale-estimation guarantee before proving it. \begin{scalerestatement}[Estimating the Disagreement Count, Restated] 
For each fixed $\delta\in(0,1)$, there is a constant $\kappa\ge1$ and a public-coin protocol using $O_\delta(\log n)$ bits with the following property. On binary strings with $k$ disagreements it outputs the same value $\widetilde k$ to both parties. If $k=0$, the output is always zero. If $k>0$, then
$$
\Prb[k/\kappa\le\widetilde k\le\kappa k]\ge1-\delta.
$$
Each nonzero output is an integer power of two.
\end{scalerestatement}

\begin{proof}
Let $\xi,\eta\in\{0,1\}^n$ be the input strings.
Let $D\coloneqq\{i:\xi_i\ne\eta_i\}$, $k\coloneqq|D|$, and $J\coloneqq\lceil\log n\rceil$. Independently choose each $U_i$ uniformly from $\{1,\ldots,2^J\}$ and set
$$
A_j\coloneqq\{i:U_i\le2^{J-j}\},\qquad 0\le j\le J.
$$
These sets are nested, and each coordinate belongs to $A_j$ with probability $2^{-j}$. Choose a positive integer $a_0$ satisfying
\begin{equation}\label{eq:scale-constants}
2^{1-a_0}+e^{-2^{a_0}}\le\delta;
\end{equation}
for example, $a_0=\lceil\log(3/\delta)\rceil$ suffices. Independently choose uniform public vectors $r_i\in\F_2^{a_0}$. At each level $j$, Alice sends $\bigoplus_{i\in A_j:\xi_i=1}r_i$, from which Bob obtains
$$
F_j\coloneqq\bigoplus_{i\in D\cap A_j}r_i.
$$
Conditional on all $U_i$, this fingerprint is zero deterministically when $D\cap A_j=\varnothing$, and otherwise is zero with probability $2^{-a_0}$ by Lemma~\ref{lem:linear-hash}. Thus
$$
\Prb[F_j=0]\le\Prb[D\cap A_j=\varnothing]+2^{-a_0}.
$$
Bob returns $\widetilde k=2^{j^*}$ for the largest level $j^*$ with $F_{j^*}\ne0$, or zero if no such level exists, and sends Alice that index or the zero marker. If $k=0$, every fingerprint is zero.

For $k>0$, define the two analysis thresholds
$$
j_-\coloneqq\max\{0,\lfloor\log k\rfloor-a_0\},\qquad j_+\coloneqq\lceil\log k\rceil+a_0.
$$
When $j_->0$, independent coordinate sampling gives
$$
\Prb[D\cap A_{j_-}=\varnothing] =(1-2^{-j_-})^k\le e^{-k2^{-j_-}}\le e^{-2^{a_0}};
$$
when $j_-=0$, this probability is zero. Except with probability $e^{-2^{a_0}}+2^{-a_0}$, therefore, $F_{j_-}\ne0$ and
$$
\widetilde k\ge2^{j_-}\ge2^{\lfloor\log k\rfloor-a_0} \ge k/2^{a_0+1}.
$$
For the upper tail, let $B_+$ be the event that some level $j\in\{0,\ldots,J\}$ with $j\ge j_+$ has $F_j\ne0$. If $j_+\le J$, nesting implies $B_+\subseteq\{D\cap A_{j_+}\ne\varnothing\}$, so a union bound over $D$ gives
$$
\Prb[B_+]\le\sum_{i\in D}\Prb[i\in A_{j_+}] =k2^{-j_+}\le2^{-a_0}.
$$
If $j_+>J$, the event is empty. Outside $B_+$, either $\widetilde k=0$ or $j^*<j_+$; hence in both cases $\widetilde k\le2^{j_+}\le2^{a_0+1}k$. A union bound and~\eqref{eq:scale-constants} prove the guarantee with $\kappa=2^{a_0+1}$. The argument uses nesting of the sampled sets, not monotonicity or independence of the fingerprints. Communication is $a_0(J+1)+\lceil\log(J+2)\rceil$ bits, and all random choices use finite sample spaces.
\end{proof}

\subsection{Recovering the Disagreement Set}\label{sec:aux-batch}

\begin{proof}[Proof of Lemma~\ref{lem:batch-recovery}]
Work over $\F_2$, let $d_0\coloneqq\xi+\eta$, and fix a public enumeration of
$$
\calD_B\coloneqq\{d\in\F_2^n:|d|\le B\}, \qquad N_B\coloneqq|\calD_B|=\sum_{j=0}^B\binom nj,
$$
where $|d|$ is Hamming weight. Choose a uniform public matrix $H\in\F_2^{m\times n}$ with $m\coloneqq\lceil\log N_B+\log(1/\varepsilon)\rceil$. Alice sends $H\xi$, allowing Bob to compute $Hd_0=H\xi+H\eta$. Bob searches $\calD_B$ for vectors with this syndrome. If there is exactly one candidate, he sends a success bit and its enumeration index; otherwise he sends a failure bit. In the success case both parties return the candidate's support.

If $|d_0|\le B$, the true vector is a candidate. By Lemma~\ref{lem:linear-hash} and a union bound, the probability of any other candidate having the same syndrome is at most
$$
(N_B-1)2^{-m}\le\varepsilon.
$$
Outside this event, both parties recover $D=\supp(d_0)$. Regardless of the input or recovery outcome, communication is at most
$$
m+\lceil\log N_B\rceil+1 \le2\log N_B+\log(1/\varepsilon)+3.
$$
Every returned candidate has weight at most $B$, even when $|d_0|>B$; this latter case need not be detected.

For completeness, let $t=B/n\in(0,1/2]$. Since $t^j\ge t^B$ for $0\le j\le B$, the binomial theorem gives
$$
t^B N_B\le\sum_{j=0}^B\binom nj t^j \le(1+t)^n\le e^{nt}=e^B.
$$
Thus $\log N_B\le B\log(en/B)=O(B\log(n/B))$, proving the cost bound. Only the finitely many bits of $H$ are random; exhaustive local search is allowed in the communication model.
\end{proof}

\subsection{Error of Sampled Counts}\label{sec:aux-count}

\begin{proof}[Proof of Lemma~\ref{lem:count-error}]
Independent coordinate sampling gives $Z_e\sim\operatorname{Bin}(T_e,p)$, hence
$$
\E\widehat T_e=T_e, \qquad \Var(\widehat T_e)=\frac{T_e(1-p)}p.
$$
For $T_e>0$, rationalizing the square-root difference and using $\widehat T_e\ge0$ yields
$$
\E(\sqrt{T_e}-\sqrt{\widehat T_e})^2 =\E\frac{(T_e-\widehat T_e)^2} {(\sqrt{T_e}+\sqrt{\widehat T_e})^2} \le\frac{\Var(\widehat T_e)}{T_e} =\frac{1-p}{p}.
$$
For $T_e=0$, both $Z_e$ and $\widehat T_e$ vanish, so the contribution is zero. Since $\widehat k\ge0$ and $k>0$, Lemma~\ref{lem:hellinger} and linearity of expectation give
$$
\begin{aligned}
\E H(J_T^E,\widehat J_E)^2 &=\E\frac{\sum_{e\in E}(\sqrt{T_e}-\sqrt{\widehat T_e})^2} {k+\widehat k}\\
&\le\frac1k\sum_{e\in E} \E(\sqrt{T_e}-\sqrt{\widehat T_e})^2 \le\frac{|E|(1-p)}{pk} \le\frac{q^2(1-p)}{pk}.
\end{aligned}
$$
The Hellinger identity includes $\widehat k=0$. When $p=1$, every count is exact. All expectations here concern the ideal coordinate sample, without conditioning on a recovery event.
\end{proof}

\subsection{The Adjacent-Integer Degree Estimate}\label{sec:aux-grid}

The following restatement specifies the boundedness and separation conditions used in the reduction to Lemma~\ref{lem:nayak-wu}.
\begin{adjacentrestatement}[Degree Needed to Separate Adjacent Integers, Restated]
Let $a,b$ be nonnegative integers. Suppose a real polynomial $F$ satisfies $|F(j)|\le M$ for all integers $-a\le j\le b+1$, and $|F(1)-F(0)|\ge\Delta>0$. Then
$$
\deg F\ge c_{M/\Delta}\sqrt{(a+1)(b+1)},
$$
where $c_{M/\Delta}>0$ depends only on the indicated ratio.
\end{adjacentrestatement}

\begin{proof}
The hypotheses imply $M\ge\Delta/2>0$. Negate $F$ if necessary so that $F(1)>F(0)$, and set
$$
t_0\coloneqq\frac{F(0)+F(1)}2,\qquad Z(t)\coloneqq\frac12+\frac{F(t)-t_0}{4M},\qquad \eta\coloneqq\frac{\Delta}{8M}.
$$
Since $|t_0|\le M$, we have $Z(j)\in[0,1]$ at every integer grid point, with $Z(0)\le1/2-\eta$ and $Z(1)\ge1/2+\eta$. Choose the smallest odd $\ell\ge3/(4\eta^2)$ and let $R\coloneqq\Maj_\ell\circ Z$. For Bernoulli trials with success probability at most $1/2-\eta$, a strict majority deviates from its mean by at least $\ell\eta$; its probability is at most $1/(4\ell\eta^2)\le1/3$ by Chebyshev's inequality, since the variance is at most $\ell/4$. The symmetric estimate applies above $1/2+\eta$. Consequently,
$$
R(j)\in[0,1],\qquad R(0)\le1/3,\qquad R(1)\ge2/3, \qquad \deg R\le\ell\deg F,
$$
where the first assertion holds throughout the integer grid and $\ell$ depends only on $M/\Delta$.

Set $L\coloneqq a+b+1$ and
$$
\widetilde R(w)\coloneqq \operatorname{ml}\bigl(R(w_1+\cdots+w_L-a)\bigr).
$$
Affine substitution and multilinearization do not increase degree, so $\deg\widetilde R\le\ell\deg F$. On the Boolean cube, $\widetilde R(w)=R(|w|-a)\in[0,1]$, because $|w|-a$ ranges from $-a$ to $b+1$. At weights $a$ and $a+1$, its values are respectively at most $1/3$ and at least $2/3$. Lemma~\ref{lem:nayak-wu} therefore gives
$$
\ell\deg F\ge\deg\widetilde R \ge c_{\mathrm{NW}}\left(\sqrt L+\sqrt{u(L-u)}\right),
$$
where $u\in\{a,a+1\}$ maximizes $|u-L/2|$. For either choice, $u(L-u)\ge ab$: the two possibilities are $a(b+1)$ and $(a+1)b$. Hence
$$
\sqrt L+\sqrt{u(L-u)} \ge\sqrt{L+ab}=\sqrt{(a+1)(b+1)}.
$$
Dividing by $\ell$ proves the result with $c_{M/\Delta}=c_{\mathrm{NW}}/\ell>0$.
\end{proof}

\subsection{Tracking the Alphabet Constants}\label{sec:aux-alphabet}

We now track the constants to prove the following estimate. 
\begin{alphabetrestatement}[Alphabet Dependence of the Present Proof, Restated]
Let $K(q)\coloneqq\max\{1,K_{q,1/(4q)}\}$, where $K_{q,1/(4q)}$ is the approximation constant in Lemma~\ref{lem:profile}. The coefficient in Theorem~\ref{thm:main} can be chosen to satisfy
$$
C_q\le 2^{O(q)}K(q)^2,
$$
where the constant in $O(q)$ is absolute.
\end{alphabetrestatement}

\begin{proof}
Let $L\coloneqq|\calE|\le2^{2q-1}-1$. In the classical protocol, $\delta=1/(1000L)$ and the explicit choice in Lemma~\ref{lem:scale} gives $\kappa=O(L)$. Consequently, $A_q=O(L^2q^2)$, $s=O(L^2q^2h_f^{-2})$, and $B=O(L^2s)$. For each of the $L$ events, recovery and transmission of the original symbols cost at most $O(B\log(n/B)+B\log(2q)+\log(1/\delta))$ in the sparse branch. The scale-estimation cost is $O(L\log(2L)\log n)$. We apply Lemma~\ref{lem:phi} as in the proof of Theorem~\ref{thm:classical} and include the cost of the full-input branch. The constant in that theorem can therefore be chosen to satisfy
$$
C_q^{\mathrm{cl}}=O\!\left(q^2L^5\log(2q)\right).
$$

In Lemma~\ref{lem:single-cycle}, the polynomial has degree at most $K(q)(c+\log(24q^3))$. The proof of Lemma~\ref{lem:adjacent-grid}, with $M=2$ and $\Delta=1/(12q^3)$, uses a majority polynomial of degree $O(q^6)$. The constant in its degree lower bound can therefore be chosen as $\Omega(q^{-6})$, with an absolute implied constant. Following~\eqref{eq:interval-arms} and the final step of Theorem~\ref{thm:inverse}, we may choose its lower-bound constant so that
$$
a_q^{-1}=O\!\left(q^7K(q)\log(2q)\right).
$$
The proof of Theorem~\ref{thm:main} then permits
$$
C_q=4(q-1)+C_q^{\mathrm{cl}}\max\{1,4a_q^{-2}\} \le O\!\left(q^{16}L^5\log^3(2q)K(q)^2\right).
$$
Since $L\le2^{2q-1}-1$, this implies~\eqref{eq:alphabet-upper}.
\end{proof}
\subsection{Comparison with the Binary Jump Measure}
\label{sec:aux-binary-comparison}

The comparison below shows that our parameter retains the binary jump measure up to absolute constants.

\begin{proposition}[Comparison with the Binary Jump Measure]
\label{prop:binary-parameter-comparison}
Let $g$ be a binary permutation-invariant partial function with fixed weights $|x|=a$ and $|y|=b$. Its feasible intersection counts form $I=\{\max\{0,a+b-n\},\ldots,\min\{a,b\}\}$, and let $\varphi(t)$ be its value at intersection count $t$. A pair $(c,d)$ with $d>0$ and $c-d,c+d\in I$ is called a \emph{jump} of $\varphi$ if its endpoints have opposite promised labels and every integer strictly between them is unpromised. Let $w_1(c)\le w_2(c)$ be the two smallest numbers, counted with multiplicity, among
$$
c,\quad a-c,\quad b-c,\quad n-a-b+c.
$$
Following~\cite{GHYY}, set
$$
m_{\mathrm{GHYY}}(g)\coloneqq \max_{(c,d)\text{ a jump of }\varphi} \frac{\sqrt{w_1(c)w_2(c)}}{d},
$$
with the maximum over an empty set equal to zero. Then
\begin{equation}\label{eq:binary-parameter-comparison}
\max\{1,m_{\mathrm{GHYY}}(g)\} \le h_g^{-1} \le\sqrt2\,\max\{1,m_{\mathrm{GHYY}}(g)\}.
\end{equation}
Here $c,d$ may be half-integers; only the two endpoints must be integers.
\end{proposition}

\begin{proof}
If at most one promised label occurs, there are no jumps and $h_g=1$, so the result follows. This includes the cases in which a margin is zero or $n$. Suppose both labels occur. For every real $t\in[\max\{0,a+b-n\},\min\{a,b\}]$, define the table
$$
T_t\coloneqq
\begin{pmatrix}
n-a-b+t & b-t\\
a-t & t
\end{pmatrix}.
$$
When $t$ is an integer, $T_t$ is the joint type with
intersection count $t$. Fix two distinct feasible counts $u<v$, set $c=(u+v)/2$ and $d=(v-u)/2$, and denote the four entries of $T_c$ by $w_e$. The endpoint entries in each cell are $w_e-d$ and $w_e+d$ in some order, so $w_e\ge d>0$. Rationalization gives
$$
\frac{d^2}{w_e} \le (\sqrt{w_e+d}-\sqrt{w_e-d})^2 =\frac{2d^2}{w_e+\sqrt{w_e^2-d^2}} \le\frac{2d^2}{w_e}.
$$
For a binary alphabet, the nonempty cut events are the two rows, two columns, two perfect matchings, and all four cells. Thus every pair of distinct cells is a cut event. For $E=\{e,e'\}$, its denominator is $2(w_e+w_{e'})$, whence
$$
\frac{d^2}{2w_ew_{e'}} \le h_E(T_u,T_v)^2 \le\frac{d^2}{w_ew_{e'}}.
$$
The all-cell ratio is a weighted average of the ratios for any two complementary pairs. It therefore cannot exceed the largest pair ratio. Writing $w_1(c),w_2(c)$ for the two smallest cell counts gives
$$
\frac{\sqrt{w_1(c)w_2(c)}}d \le h(T_u,T_v)^{-1} \le\sqrt2\,\frac{\sqrt{w_1(c)w_2(c)}}d.
$$

It remains to compare arbitrary opposite-label pairs with jumps. Every interval $[u,v]$ whose endpoints have opposite labels contains two consecutive promised counts $u'<v'$ with opposite labels. They form a jump. Set $c'=(u'+v')/2$ and $d'=(v'-u')/2$. Interval inclusion implies $d'\le d$ and $|c'-c|\le d-d'$. Since every entry of $T_t$ has slope $+1$ or $-1$, for every cell $e$,
$$
(T_{c'})_e\ge w_e-(d-d')\ge\frac{d'}d\,w_e,
$$
where the last inequality uses $w_e\ge d$. The first and second smallest entries therefore satisfy the same componentwise bounds, and hence
$$
\frac{\sqrt{w_1(c')w_2(c')}}{d'} \ge\frac{\sqrt{w_1(c)w_2(c)}}d.
$$
The maximum of this score over all opposite-label pairs consequently equals its maximum over jumps. Since $h_g^{-1}$ is the maximum of $h(T_u,T_v)^{-1}$ over opposite-label pairs, the preceding pairwise bounds prove $m_{\mathrm{GHYY}}(g)\le h_g^{-1}\le\sqrt2\,m_{\mathrm{GHYY}}(g)$. Finally, $w_1(c),w_2(c)\ge d$ imply that every jump has score at least one, giving the stated formula in all cases.
\end{proof}

\clearpage
\phantomsection
\addcontentsline{toc}{section}{References}
\bibliographystyle{alpha}
\bibliography{references}
\end{document}